\documentclass[]{IEEEtran}
\usepackage{amssymb,amsmath,bm,epsfig,graphicx,theorem,epstopdf,wasysym}
\usepackage{rotating,setspace,latexsym,epsf,color,dsfont,caption}
\usepackage{cite,subfigure}

\allowdisplaybreaks

\newtheorem{Theorem}{Theorem}

\newtheorem{Lemma}{Lemma}

\newtheorem{Definition}{Definition}

\newenvironment{Proof}[1]{\medskip\par\noindent
{\bf Proof:\,}\,#1}{{\mbox{\,$\blacksquare$}\par}}
\newenvironment{proof}[1]{\medskip\par\noindent
{\bf Proof:\,}\,#1}{{\mbox{\,$\blacksquare$}\par}}

\newcommand{\Rb}{\mathbb{R}}
\newcommand{\Eb}{\mathbb{E}}
\newcommand{\Pb}{\mathbb{P}}

\newcommand{\lv}{\mathbf{1}}

\newcommand{\tcb}{\textcolor{black}}

\begin{document}
\title{Age of Information Minimization for an Energy Harvesting Source with Updating Erasures:\\ Without and With Feedback}
\author{Songtao~Feng and Jing~Yang%
\thanks{Songtao~Feng and Jing~Yang are with the School of Electrical Engineering and Computer Science, The Pennsylvania State University, University Park, PA, 16802, USA. Email: \{sxf302, yangjing\}@psu.edu. This work is presented in part in the 2018 IEEE International Conference on Computer and Communications (INFOCOM) - Workshop on Age of Information~\cite{Feng:2018:INFOCOM} and the 2018 IEEE International Symposium on Information Theory~\cite{Feng:2018:ISIT}.}}%

%


\maketitle

\begin{abstract}
Consider an energy harvesting (EH) sensor that continuously monitors a system and sends time-stamped status update to a destination.
The sensor harvests energy from nature and uses it to power its updating operations.
The destination keeps track of the system status through the successfully received updates.
With the recently introduced information freshness metric ``Age of Information" (AoI), our objective is to design optimal online status updating policy to minimize the long-term average AoI at the destination, subject to the energy causality constraint at the sensor.
Due to the noisy channel between the sensor and the destination, each transmitted update may be erased with a fixed probability, and the AoI at the destination will be reset to zero only when an update is successfully received.
We first consider status updating without feedback available to the sensor and show that the Best-effort Uniform updating (BU) policy is optimal \tcb{in a broadly defined class of online policies}. We then investigate status updating with perfect feedback to the sensor and prove \tcb{similar} optimality of the Best-effort Uniform updating with Retransmission (BUR) policy. In order to prove the optimality of the proposed policies, for each case, we first identify a lower bound on the long-term average AoI among a broad class of online policies, and then construct a sequence of virtual policies to approach the lower bound asymptotically. Since those virtual policies are sub-optimal to the original policy, the original policy is thus optimal.
\end{abstract}



\section{Introduction}
Recently, a metric called ``Age of Information" (AoI) has been introduced to measure the freshness of the information in a status monitoring system from the destination's perspective ~\cite{infocom/KaulYG12}.
Specifically, at time $t$, the AoI in the system is defined as $t-U(t)$, where $U(t)$ is the time stamp of the latest received update at the destination.
AoI has shown to be fundamentally different from standard network performance metrics, such as throughput or delay.
It has attracted growing attention from different research communities, due to its simple form and potential in unifying sampling and transmission for timely information delivery.

Generally speaking, there are two main approaches in the study of AoI. The first approach is to characterize the AoI under given status updating policies~\cite{infocom/KaulYG12,ciss/KaulYG12,isit/YatesK12,YatesK16,Pappas:2015:ICC,isit/NajmN16,isit/KamKNWE16,isit/ChenH16,isit/KamKE13, isit/KamKE14,tit/KamKNE16,isit/CostaCE14,tit/CostaCE16,isit/HuangM15,Yates:2018:AoISHS}. The second approach is to design certain status updating policies to actively optimize AoI~\cite{isit/BedewySS16,infocom/SunUYKS16,Sun:ISIT:2017}.
Modeling the status monitoring system as a queueing system, where updates are generated at the source according to a random process, the time average AoI has been analyzed in different queueing management settings. For systems with a single server, the corresponding AoI has been studied in single-source single-server queues~\cite{infocom/KaulYG12}, the $M/M/1$ Last-Come First-Served (LCFS) queue with preemption in service~\cite{ciss/KaulYG12}, the $M/M/1$ First-Come First-Served (FCFS) queue with multiple sources~\cite{isit/YatesK12,YatesK16}, the $M/M/1$ queue with multiple souces which only keeps the latest status packet of each source in the queue~\cite{Pappas:2015:ICC}, the LCFS queue with gamma-distributed service time and Poisson update packet arrivals~\cite{isit/NajmN16}. Moreover, in $M/M/1$ queue systems, packet deadlines are found to improve AoI performance in~\cite{isit/KamKNWE16}, and AoI in the presence of packet delivery errors is evaluated in~\cite{isit/ChenH16}. The AoI in systems with multiple servers has been evaluated in~\cite{isit/KamKE13, isit/KamKE14,tit/KamKNE16}. A related metric, Peak Age of Information (PAoI), is introduced and studied in~\cite{isit/CostaCE14,tit/CostaCE16,isit/HuangM15,infocom/SunUYKS16}. 
For more complicated multi-hop networks, reference~\cite{Yates:2018:AoISHS} introduces a novel stochastic hybrid system (SHS) approach to derive explicit age distributions.
The optimality properties of a preemptive Last Generated First Served (LGFS) service discipline in a multi-hop network are identified in~\cite{isit/BedewySS16}. AoI optimization with the knowledge of the server state has been studied in~\cite{infocom/SunUYKS16}. The relationship between AoI and the MMSE in remote estimation of a Wiener process is investigated in~\cite{Sun:ISIT:2017}.

\tcb{Age of information has also demonstrated its fundamental role in the state estimation and real-time control of stochastic systems. In \cite{Baras:AoI}, the fundamental trade-off between the control performance and information staleness measured in AoI has been characterized, while
in \cite{ZhangJia:AoI:ISIT2018}, it studies how the random AoI would alter the rate-cost tradeoff for a Gaussian linear control system, where the cost is measured in terms of the system-state mean-square stability. In \cite{Jaya:AoI}, AoI has been adopted to solve the state estimation and control problem in a single-loop stochastic linear time-invariant (LTI) networked system. It shows that that minimizing the estimation error is equivalent to minimizing a non-negative and non-decreasing function of AoI. In \cite{Aritra:AoI:LTI}, AoI has been utilized for the distributed estimation of the state of a discrete-time LTI process over a time-varying directed communication graph.}


Due to the magnified tension between keeping information fresh and the stringent energy constraint, AoI in energy harvesting (EH) wireless networks has attracted increasing interests recently~\cite{isit/Yates15,ita/BacinogluCU15,Yang:2017:AoI,Arafa:AoI:finitebatt,BacinogluU17,Uysal:2018:EH,Ahmed:2017:Asilomar,Ahmed:2017:Globecom,Baknina:2018:CISS,Baknina:2018:ISIT,Brown:2018:INFOCOM}. An EH sensor harvests energy from the environment and uses it to power its sensing and communication operations. Due to the stochastic energy arrival process, all of the operations are subject to the so-called energy causality constraint. Under such constraints, various policies have been proposed to optimize different communication and sensing performance metrics~\cite{Yang_tcom,ozel_finite_tcom,Huang0C13,HoZ12,kaya_tcom, Ozel0U13,Yang_TWC_broadcast,uysal_paper,Lai:quickest,Yang:jsac:2016,Yang:jsac:2015}. Such sample path-wise constraint also makes the design and analysis of the status updating policy in EH systems extremely challenging.
Under the assumption that the battery size is sufficiently large, \cite{isit/Yates15} shows that updates should be scheduled only when the server is free to avoid queueing delay, and a {\it lazy} update policy that introduces inter-update delays outperforms the greedy policy.
Reference~\cite{ita/BacinogluCU15} investigates AoI-optimal offline and online status updating policies, where the online problem is modeled as a Markov decision process and solved through dynamic programming.
In~\cite{Yang:2017:AoI,Arafa:AoI:finitebatt,BacinogluU17,Uysal:2018:EH}, optimal online status updating policies under different assumptions on the battery size have been identified. Specifically, for the infinite battery case, \cite{Yang:2017:AoI} shows that the 
{\it best-effort uniform} updating policy, which updates at a constant rate when the source has sufficient energy, is optimal when the channel between source and destination is perfect.
When the battery size is finite, the optimal policies are shown to have certain threshold structures~\cite{Arafa:AoI:finitebatt,BacinogluU17,Uysal:2018:EH}. Offline policies to minimize AoI in EH channels have been studied in~\cite{Ahmed:2017:Asilomar,Ahmed:2017:Globecom}.
Reference~\cite{Baknina:2018:CISS} analyzes the AoI performance of two channel coding schemes when channel erasures are present.
Using the SHS tools proposed in~\cite{Yates:2018:AoISHS}, reference~\cite{Brown:2018:INFOCOM} and reference \cite{Brown:2018:ISIT} study the average AoI for a finite battery EH system, with and without preemption of packets in service allowed, respectively.
An interesting setting is considered in~\cite{Baknina:2018:ISIT}, where extra information is carried by the timing of the update packets. A tradeoff between the average AoI and the average message rate is studied for several achievable schemes.

In this paper, we take the imperfect updating channel into consideration and investigate the optimal updating policies of an EH system where updating erasures can happen.
Assuming each update can be erased with a constant probability, the AoI at the destination will be reset only when an update is successfully received. Our objective is to design \emph{online} status updating policies to minimize the average AoI at the destination.
Depending on whether there exists updating feedback to the source, we consider two possible scenarios:

1) \emph{No updating feedback}.
In this case, the source has no knowledge of whether an update is successful. It can only use the update-to-date energy arrival profile and updating decisions as well as the statistical information, such as the energy arrival rate and the erasure probability of the channel, to decide the upcoming updating time points.
We show that the {\it Best-effort Uniform updating} (BU) policy, which was shown to be optimal under the perfect channel setting in~\cite{Yang:2017:AoI}, is still optimal \tcb{among a broad class of online policies}.

2) \emph{Perfect updating feedback}.
In this case, the source receives an instantaneous feedback when an update is transmitted.
Therefore, it can decide when to update next based on the feedback information, along with the information it uses for the no feedback case.
For this case, we propose a {\it Best-effort Uniform updating with Retransmission} (BUR) policy and prove its optimality \tcb{among a broad class of online policies}. 

Although the proposed policies are quite intuitive, their optimality is quite challenging to establish, compared with~\cite{Yang:2017:AoI}.
This is because both battery outage and updating erasure will affect the AoI under the proposed policies.
While the impact of either of those two events can be analyzed relatively easily when isolated, it becomes extremely challenging when both of them are involved.
Besides, when there exists perfect updating feedback to the source, updating erasures under the BUR will lead to subsequent retransmissions and energy consumption, thus affecting the battery outage probability in the future.
Such complicated interplay between those two events makes the problem even more complicated.
In order to overcome such difficulties, we propose a novel {\it virtual policy} based approach.
Specifically, for both BU and BUR updating policies, we construct a sequence of virtual policies, which are strictly suboptimal to their original counterparts, and eventually converge to them.
\tcb{Leveraging the virtual policies,} 
we are able to decouple the effects of battery outage and updating errors in the performance analysis.
We show that the long-term average AoI under virtual policies converges to the corresponding lower bound, which implies the optimality of the original policy.

The remainder of the paper is structured as follows:
In Sec.~\ref{sec:model}, we describe the system model and problem formulation.
In Sec.~\ref{sec:no-feedback} and Sec.~\ref{sec:feedback}, we consider the no updating feedback case and the perfect updating feedback case, respectively. 
In Sec.~\ref{sec:simulation}, we evaluate the proposed policies through extensive simulation results.
We conclude in Sec.~\ref{sec:conclusion}. For the sake of readability, we defer some proofs to the appendix.

\section{System Model and Problem Formulation} \label{sec:model}
Consider a scenario where an energy harvesting sensor continuously monitors a system and sends time-stamped status updates to a destination.
The destination keeps track of system status through received updates.
We use the metric Age of Information (AoI) to measure the ``freshness" of the status information available at the destination.

We assume that the energy unit is normalized so that each status update requires one unit of energy.
This energy unit represents the energy cost of both measuring and transmitting a status update.
Assume energy arrives at the sensor according to a Poisson process with parameter $\lambda$. Hence, energy arrivals occur at discrete time instants $t_1,t_2,\ldots$.
We assume $\lambda=1$ for ease of exposition, since we can always scale the time axis proportionally to make $\lambda=1$ per unit time.
The sensor is equipped with a battery to store harvested energy.
In this paper, we focus on the case where the battery size is infinite.

We assume \tcb{that} the time used to collect and transmit a status update is {\it negligible} compared with the time scale of the long-term average AoI in the system. Therefore, a status update can be generated and transmitted at any time as long as the energy level is greater than or equal to one. We assume \tcb{that} the channel between the source and the destination is \tcb{time-invariant and noisy, thus with probability $1-p$, $0<p\leq 1$, each update will be erased during transmission}, independent \tcb{of} any other factors in the system. As shown in Fig.~\ref{fig:AoI}, the AoI at the destination will be reset to zero only when an update is successfully received. We consider two possible cases. For the no updating feedback case, the source has no information of the updating result. For the perfect updating feedback case, we assume there is a {\it perfect feedback} channel between the destination and the source, so that the source is notified about an updating failure once it happens. 

A status update policy is denoted as $\pi:=\{l_n\}_{n=1}^\infty$, where $l_n$ is the $n$th update time at the {\it source}. However, due to \tcb{random update erasures}, only a subset of the update packets will be successfully delivered. Thus, the actual status update times at the {\it destination} are different from $\{l_n\}_{n=1}^\infty$ in general. Therefore, we use $S_n$ to denote the $n$th actual update time at the {\it destination}.   
We assume $S_0=l_0=0$, i.e., an update is successfully delivered right before time zero, and the system starts with an initial energy of $E_0$, $E_0\geq 1$. 

\if{0}
\begin{figure}[t]
	\centering
	\includegraphics[scale=0.6]{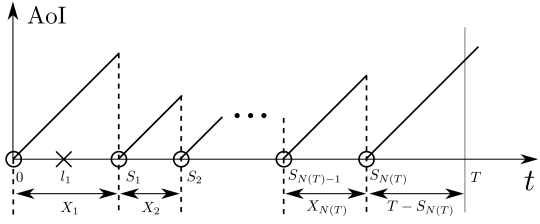}
	\caption{AoI as a function of $T$. Circles $\circ$ represent successful status updates, and crosses $\times$ represent failed status updates.}\label{fig:AoI}
\end{figure}
\fi

\begin{figure}[t]
	\centering
	\includegraphics[scale=0.28]{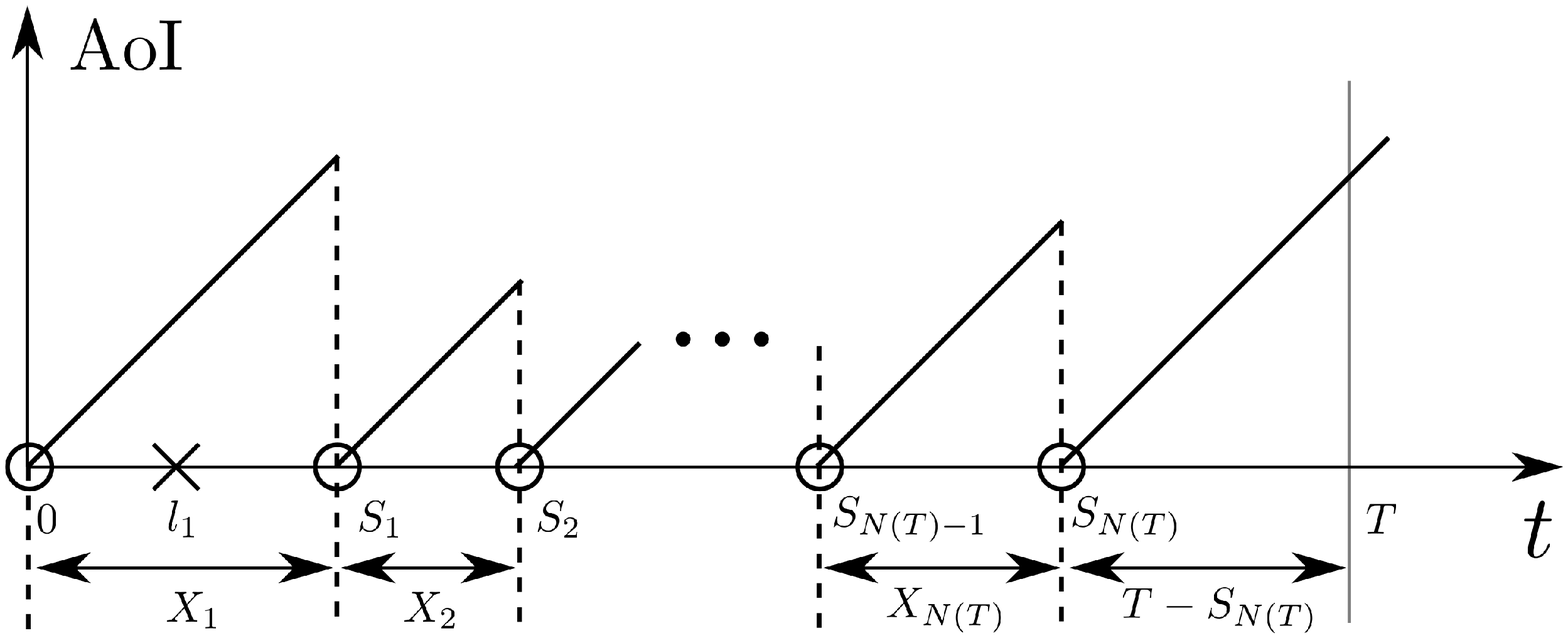}
	\caption{{AoI as a function of $t$. Circles $\circ$ represent successful status updates, and crosses $\times$ represent failed status updates.}}\label{fig:AoI}
	\vspace{-0.1in}
\end{figure}

Define $A_n$ as the total amount of energy harvested in $[l_{n-1},l_n)$, and $E(l^-_n)$ as the energy level of the sensor right before the update time $l_n$. Then, under any feasible status update policy, the energy queue evolves as follows
\begin{align}
E(l_1^-)&=E_0+A_1, \label{eqn:energy_initial}\\
E(l^-_{n})&=E(l^-_{n-1})-1+A_n, \quad n=2,3,\ldots.\label{eqn:energy_queue}
\end{align}
Based on the Poisson arrival process assumption, $A_n$ is an independent Poisson random variable with parameter $l_{n}-l_{n-1}$. 

In order to ensure every update time is feasible, we must have the energy causality constraint satisfied all the time, i.e.,
\begin{align}
E(l_n^-)&\geq 1, \quad n=1,2,\ldots, \label{eqn:energy_constraint}
\end{align}
which indicates that the source will generate and transmit an update only when it has sufficient energy. 

We use $M(T)$ and $N(T)$ to denote the number of status updates sent by the source and the number of status updates successfully received at the destination over $(0,T]$, respectively. Define $R(T)$ as the cumulative AoI at the destination over $[0,T]$. Denote the delay between two successful updates as $X_n:= S_n-S_{n-1}$, for $n=1,2,\ldots$. 
Then,
\begin{align}\label{defn:R(T)}
R(T)&=\frac{\sum_{i=1}^{N(T)}X_i^2+ (T-S_{N(T)})^2}{2},
\end{align}
\tcb{which corresponds to the area below the AoI curve over $[0,T]$, as shown in Fig.~\ref{fig:AoI}.} The time-average AoI over the duration $[0,T]$ can then be expressed as $R(T)/T$.

Our objective is to determine the sequence of update times $l_1,l_2,\ldots$ at the {\it source}, so that the time average AoI at the {\it destination} is minimized, subject to the energy causality constraint. We focus on a set of {\it online} policies. Specifically, for the no updating feedback case, the information available for determining the updating point $l_n$ includes the updating history $\{l_i\}_{i=0}^{n-1}$, the energy arrival profile over $[0,l_n)$, as well as the energy harvesting statistics (i.e., $\lambda$ in this scenario) and the probability of updating success $p$. Denote the set of such online policies as $\Pi_1$. For the perfect updating feedback case, the source also utilizes up-to-date updating feedback to make its decisions. We denote the set of such online policies as $\Pi_2$.
Then, the optimization problem can be formulated as
\begin{eqnarray}\label{eqn:opt}
\underset{\pi\in\Pi}{\min} & &\limsup_{T\rightarrow +\infty} \Eb\left[\frac{R(T)}{T}\right]\\
\mbox{s.t. } & & (\ref{eqn:energy_initial})-(\ref{eqn:energy_constraint}),\nonumber
\end{eqnarray}
where $\Pi$ equals $\Pi_1$ or $\Pi_2$, depending on the setting, and the expectation in the objective function is taken over all possible energy harvesting sample paths and \tcb{update erasure patterns}.

\section{Status Updating Without Feedback} \label{sec:no-feedback}
In this section, we will study the optimal status updating policy for the case where there is no update feedback available to the sensor. We show that the expected long-term average AoI has a lower bound for a broad class of online policies, which can be achieved by the BU updating policy.

\subsection{A Lower Bound}\label{subsec:no-feedback:lower}
Note that when battery size is infinite, no energy flow will happen, and the long-term average status updating rate is subject to the energy harvesting rate constraint. Specifically, we have the following lemma.

\begin{Lemma}[Lemma 1 in \cite{Yang:jsac:2016}]\label{1lemma:rate}
Under any policy $\pi\in\Pi_1$, it must have $\limsup_{T\rightarrow \infty} M(T)/T \leq 1$ almost surely.
\end{Lemma}
We point out that Lemma~\ref{1lemma:rate} is also valid for all $\pi\in\Pi_2$, which will be discussed in Sec.~\ref{sec:feedback}.

Besides, we also have the following intuitive yet important observation.

\begin{Lemma}\label{1lemma:finiteAoI}
For any $\pi\in\Pi_1$ that achieves a {\it finite} expected long-term average AoI, it must have $\lim_{T\rightarrow \infty} M(T)=\infty$ almost surely.
\end{Lemma}

\begin{proof}
We prove it by contradiction. Assume $$\Pb\left[\lim_{T\rightarrow \infty} M(T)=\infty\right]<1,$$
i.e., there exists $\epsilon>0$ and $M_0>0$, such that 
$$\Pb\left[\lim_{T\rightarrow \infty} M(T)<M_0\right]\geq \epsilon.$$
Define
\begin{align} \label{defn:p_n}
p_n:=(1-p)^{n-1}p,
\end{align}
i.e., the probability that $l_n$ is the first successful update time after $l_0$.
Then,
\begin{align}
&\limsup_{T\rightarrow \infty} \Eb\left[\frac{R(T)}{T}\right] \nonumber\\
&\geq \lim_{T\rightarrow \infty} \frac{T^2}{2T}\cdot\Pb[\mbox{all $M(T)$ updates fail}, M(T)<M_0]  \\
&\geq \lim_{T\rightarrow \infty} \frac{T}{2}\left(1-\sum_{i=1}^{M_0}p_i\right) \epsilon=\infty    ,
\end{align}
which implies that the expected long-term average AoI cannot be finite. 
\end{proof}

In order to obtain a valid lower bound, in the following, we only need to focus on the policies that achieve finite expected long-term average AoI. To facilitate the following analysis, we introduce a broad class of online policies defined as follows.

\begin{Definition}[Bounded Updating Policy]
If under a policy $\pi\in\Pi_1$, the $n$th updating point at the source (i.e., $l_n$) satisfies $\Eb[l_n]<\infty$ for any fixed $n\in\{1,2,\ldots\}$, $\pi$ is called a {\it bounded updating} policy.
\end{Definition}

Denote the set of bounded updating policies as $\Pi_3$. Then, $\Pi_3\subset\Pi_1$. Intuitively, any practical status updating policy should be in $\Pi_3$, as it is undesirable to have any $n$th updating point (and the inter-update delay between any consecutive updating points before $l_n$) to become unbounded {\it in expectation}. We have the following lower bound for bounded updating policies.

\begin{Theorem}[Lower Bound for Channel without Feedback]\label{1thm:lower}
For any policy $\pi\in\Pi_3$, the expected long-term average AoI is lower bounded by $\frac{2-p}{2p}$.
\end{Theorem}
The proof of Theorem \ref{1thm:lower} is provided in Appendix \ref{app:1thm:lower}.

\subsection{Optimal Online Status Updating}
In this section, we propose online status updating policies to achieve the lower bound derived in Section~\ref{subsec:no-feedback:lower}. We will start with the BU updating policy introduced in \cite{Yang:2017:AoI}. Although we assume a noisy channel in this work, when there is no feedback available to the source, intuitively, it is still desirable for the source to update in a uniform fashion, so that the successfully received updates at the destination would be most uniformly distributed in time. 

\begin{Definition}[BU Updating]
The sensor is scheduled to update the status at $s_n=n$, $n=1,2,\ldots$. The sensor performs the task at $s_n$ if $E(s^-_n)\geq 1$; Otherwise, the sensor keeps silent until the next scheduled status updating time point.
\end{Definition}
Here we use $s_n$ to denote the $n$th scheduled updating time point. It is in general different from the $n$th actual updating time $l_n$, since some scheduled updates may be infeasible due to battery outage.

BU updating ensures that the energy causality constraint is always satisfied. We expect that BU updating achieves the lower bound in Theorem~\ref{1thm:lower}. However, analyzing its AoI performance is very challenging. Although we are able to identify a renewal structure in the system status evolution under the BU updating policy (i.e., a renewal interval can begin right after the sensor successfully delivers an update and the battery state becomes $E_0-1$), the analysis of the expected average AoI over one renewal interval is still very complicated, mainly due to two reasons: 

First, different from the perfect channel case~\cite{Yang:2017:AoI}, the actual update time at the destination $S_n$ may deviate from the scheduled update time $s_n$ due to two possible events: battery outage and update erasure. Although the average AoI can be characterized in systems where only one of such events can happen, it is hard to analyze the AoI when the effects of both events are involved. 

Second, the expected length of such a renewal interval is unbounded. This is because the battery evolution under BU updating can be modeled as a Martingale process, and as we will show in the proof of Lemma~\ref{1lemma:N(T1)}, the expected time when it becomes empty for the first time (i.e., hitting time of zero) is infinity. Since with a non-zero probability the renewal interval contains such an interval, the expected length of each renewal interval is thus unbounded, and the corresponding expected average AoI becomes intractable.

To overcome such challenges, we will construct a sequence of {\it virtual} policies, and show that the expected time average AoI under those virtual policies approaches the lower bound in Theorem~\ref{1thm:lower}. Since such virtual policies are sub-optimal to the BU updating policy, the optimality of BU updating can thus be proved. In order to simplify the definition and analysis of the virtual policy, we assume $E_0=1$. The proof can be slightly modified to show that the optimality of the proposed policy is valid for any $E_0\geq 0$.

\begin{Definition}[BU-ER$_{T_0}$] \label{def:bu-er}
	The sensor performs BU updating until the battery level after sending an updating becomes zero for the first time, or until time $T_0^+$, in which case the sensor depletes its battery; After that, when the battery level becomes higher than or equal to one after a successful update for the first time, the sensor reduces the battery level to one, and then repeats the process.
\end{Definition}

\begin{Lemma} \label{1lemma:suboptimal}
	For any $T_0>0$, BU-ER$_{T_0}$ updating policy is sub-optimal to the BU updating policy. 
\end{Lemma} 
\begin{proof}
	We note that BU-ER$_{T_0}$ updating is identical to BU updating except the energy removal at time $T_0$ and when $E(s_n^+)$ becomes higher than one. Given the same energy harvesting sample path, the battery level under BU is always higher than that under BU-ER$_{T_0}$. Thus, BU-ER$_{T_0}$ incurs more infeasible status updates. With the same \tcb{update erasure pattern}, the instantaneous AoI under BU-ER$_{T_0}$ updating is always greater than or equal to that under BU updating sample path-wisely. Thus, the expected time-average AoI under BU-ER$_{T_0}$ is greater than or equal to that under BU, which proves the lemma. 
\end{proof}

We note that the BU-ER$_{T_0}$ updating policy is a {\it renewal} type policy, i.e., the states of the system evolve according to a renewal process. \tcb{To see this, we note that the updating process under BU-ER$_{T_0}$ works in cycles, where each cycle begins with the initial battery level to be one and the AoI to be zero, followed by i.i.d. battery and AoI evolution processes.} Therefore, to analyze the expected long-term average AoI, it suffices to analyze the expected average AoI over one renewal interval. In the following, we will focus on the first renewal interval, and show that the corresponding expected average AoI converges to the lower bound in Theorem~\ref{1thm:lower} as $T_0$ increases. As illustrated in Fig.~\ref{fig:bu-er}, the renewal interval consists of two stages. The first stage starts at time zero and ends until the battery becomes empty for the first time, or until time $T_0^+$. We denote $T_1$ as the end of the first stage. We note that all scheduled status updating epochs over $(0, T_1]$ are feasible. The second stage starts at $T_1$ and ends when the battery level becomes higher than or equal to one after a successful update for the first time after $T_1$. We denote \tcb{the duration of the second stage as $T_2$. The second stage thus ends at $T_1+T_2$}.

\if{0}
\begin{figure}[t]\centering
\includegraphics[width=3.3in]{BUER-08012018_v2.eps}
\vspace{-0.05in}
\caption{\tcb{An illustration of the BU-ER$_{T_0}$ updating policy and the battery level right after each updating epoch. AoI will be reset to zero at the successul updating epochs.}}\label{fig:bu-er}
\vspace{-0.1in}
\end{figure}
\fi

\begin{figure}[t]\centering
\includegraphics[width=3.3in]{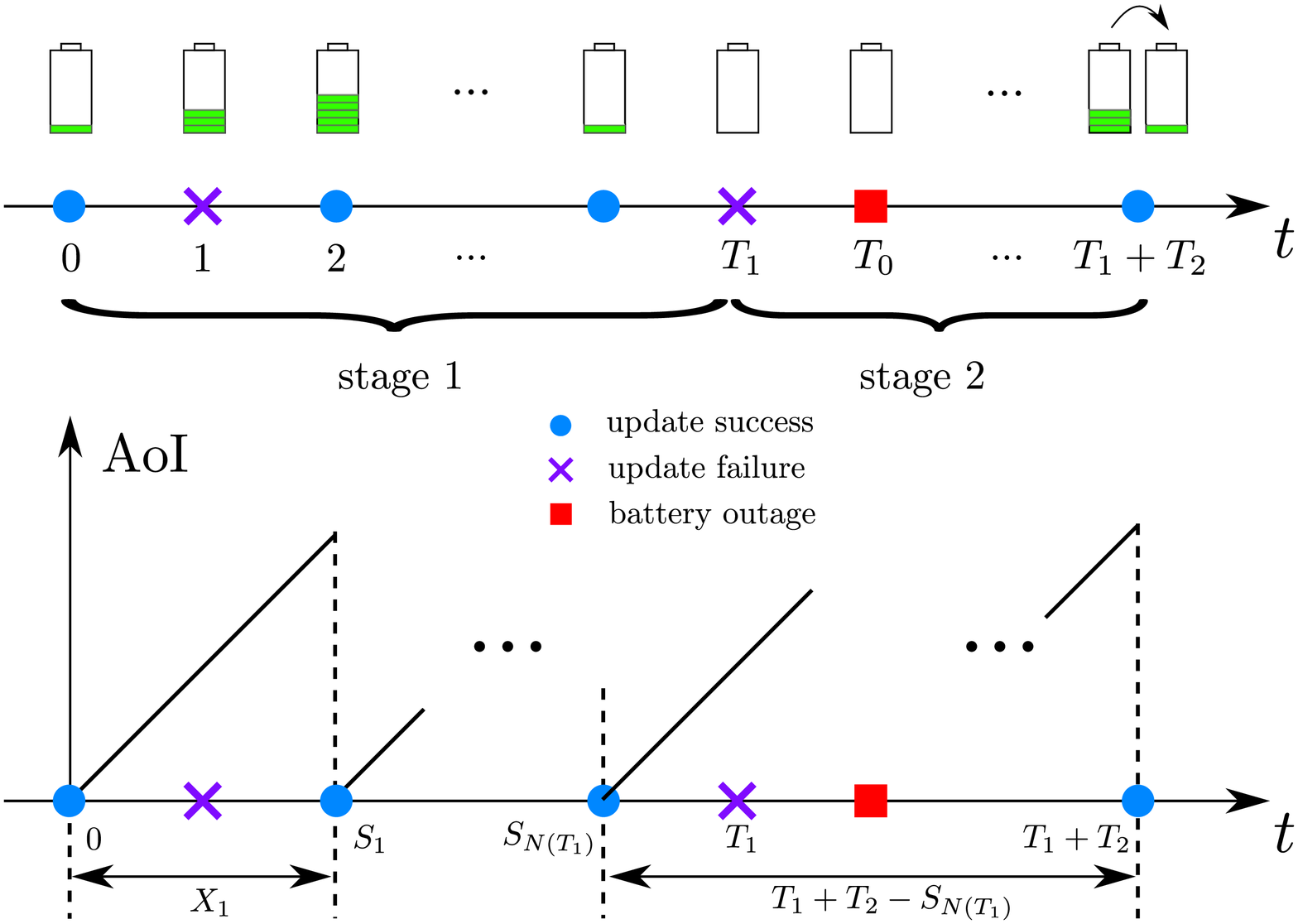}
\caption{\tcb{An illustration of the BU-ER$_{T_0}$ updating policy and the battery level right after each updating epoch. AoI will be reset to zero at the successul updating epochs.}}\label{fig:bu-er}
\vspace{-0.1in}
\end{figure}

\begin{Lemma} \label{1lemma:N(T1)}
	Under BU-ER$_{T_0}$ updating, $\lim_{T_0\rightarrow \infty} \Eb[T_1]=\infty.$
\end{Lemma}
\begin{proof}
	Consider a ``random walk" $\{\Omega_n\}_{n=0}^{\infty}$, which starts with $1$ and increments with $A_n-1$, where $A_n$ is an i.i.d. Poisson random variable with parameter $1$. Denote the first $0$-hitting time for $\{\Omega_n\}_{n=0}^{\infty}$ as $\kappa$. Then $\Omega_0=1$ and $\Omega_{\kappa}=0$. Note that when $T_0\rightarrow\infty$, $\{\Omega_n\}_{n=0}^{\kappa}$ is identical to the battery level evolution process $\{E(s_n^+)\}_{n=0}^{\kappa}$ under the BU-ER$_{T_0}$ updating policy almost surely, and the corresponding $T_1=\kappa$. 
	
	Define a Martingale process associated with $\{\Omega_n\}_{n=0}^{\infty}$ as $\{\exp(-\alpha \Omega_n - n\gamma(\alpha))\}_{n=0}^{\infty}$ with $\alpha>0$ and $\gamma(\alpha)=e^{-\alpha}-(1-\alpha)>0$. According to the proof of Theorem $1$ in \cite{Yang:jsac:2016}, 
	\begin{align} \label{eqn:martingale-1}
	\exp(-\alpha \Omega_0) = \Eb[\exp(-\alpha \Omega_{\kappa} - \kappa\gamma(\alpha))].
	\end{align}
	Taking the derivative of both sides of (\ref{eqn:martingale-1}) with respect to $\alpha$, we have
	\begin{align} \label{eqn:martingale-2}
	\Omega_0 \exp(-\alpha \Omega_0) 
	= \Eb[(\Omega_{\kappa}+\kappa\gamma^{\prime}(\alpha))\exp(-\alpha \Omega_{\kappa} - \kappa\gamma(\alpha) )]. 
	\end{align}
	Since $\Omega_0=1$ and $\Omega_{\kappa}=0$,  (\ref{eqn:martingale-2}) can be reduced to
	\begin{align} \label{eqn:martingale-3}
	\exp(-\alpha)
	=\Eb[\kappa \gamma^{\prime}(\alpha) \exp(-\kappa \gamma(\alpha) )]\leq \Eb[\kappa\gamma^{\prime}(\alpha)],
	\end{align}
	where the inequality follows from the fact that $\kappa\gamma(\alpha) \geq 0$.
	
	Dividing both sides of (\ref{eqn:martingale-3}) by $\gamma^{\prime}(\alpha) $, we have
	\begin{align}
	\Eb[\kappa]\geq \exp(-\alpha)/\gamma^{\prime}(\alpha). \label{eqn:martingale-revision1}
	\end{align}
	Note that 
	\begin{align}
	\lim_{\alpha \rightarrow 0} \gamma^{\prime}(\alpha)&=\lim_{\alpha \rightarrow 0}(-e^{-\alpha}+1)=0^{+}.
	\end{align} 
	\tcb{Combining (\ref{eqn:martingale-revision1}) and the fact that $T_1=\kappa$ when $T_0\rightarrow \infty$, we have}
	\begin{align}
	\lim_{T_0\rightarrow \infty}\Eb[T_1] \geq \lim_{\alpha \rightarrow 0} \exp(-\alpha)/\gamma^{\prime}(\alpha) =\infty.
	\end{align}
\end{proof}

\begin{Lemma}\label{1lemma:bounded}
	Under BU-ER$_{T_0}$ updating, $\Eb[T_2]$, $\Eb[T^2_2]$, $\Eb[T_1-S_{N(T_1)}]$, $\Eb[(T_1-S_{N(T_1)})^2]$ are bounded.
\end{Lemma}

\begin{proof}
We consider another genie-aided virtual process starting at time $T_1$ as follows. The source performs BU-ER$_{T_0}$ after $T_1$, and keeps tracking the battery level and genie-informed update result. If a status update is erased and the battery level is above zero, the sensor depletes its battery and repeat the process. The process stops when the battery level after a successful update becomes one for the first time. Denote the duration of the second state as $T_2^\prime$.

For each sample path, we can see that the battery level under the new virtual process is always less than or equal to that under BU-ER$_{T_0}$, due to the extra energy depletion after $T_1$ and before $T_2^\prime$. Since the update erasure patterns are the same under both policies, we must have $T_2^\prime>T_2$. We note that at each updating time point between $T_1$ and $T_2^\prime$, the battery level is above zero with probability $1-2e^{-1}$; and if the previous event happens, the update is successfully delivered with probability $p$. Therefore, $T_2^\prime$ under the new virtual policy is a geometric random variable with parameter $p(1-2e^{-1})$. Thus, its first and second moments are bounded. Therefore, \tcb{$\Eb[T_2]$} and $\Eb[T_2^2]$ are bounded.

Next, we note that under the BU-ER$_{T_0}$ updating, the AoI over $[0,T_1]$ is a renewal reward process, which resets to zero at $\{S_i\}_{i=1}^{N(T_1)}$. According to Proposition 3.4.6 in~\cite{ross:1996}, $\lim_{t\rightarrow\infty}\Eb[S_{N(t)}-t]$ is bounded. Therefore $\Eb[S_{N(T_1)}-T_1]$ is uniformly bounded for any $T_1$. Similarly, we can show that $\Eb[(S_{N(T_1)}-T_1)^2]$ is uniformly bounded.
\end{proof}

\if{0}
\begin{figure}[t]\centering
	\includegraphics[width=2.7in]{lemma6_pic.eps}
	\vspace{-0.05in}
	\caption{\tcb{AoI as a function of T of the BU-ER$_{T_0}$.}}\label{fig:lemma_6}
	\vspace{-0.1in}
\end{figure}
\fi

\begin{Lemma}\label{1lemma:AoI}
	As $T_0\rightarrow\infty$, the expected long-term average AoI under BU-ER$_{T_0}$ is upper bounded by $\frac{2-p}{2p}$.
\end{Lemma}
\begin{proof}
	First, we note that the 
	\begin{align}
	&\lim_{T_0\rightarrow \infty}\frac{\Eb[ (T_1+T_2-S_{N(T_1)})^2 ]}{2\Eb[T_1+T_2]} \nonumber \\
	&= \lim_{T_0\rightarrow \infty} \frac{\Eb[(T_1-S_{N(T_1)})^2] +
		\Eb[T_2^2] +2 \Eb[T_1-S_{N(T_1)}]\Eb[T_2] }{2\Eb[T_1]}  \nonumber\\
	&=0, \label{eqn:usebounded-2} 
	\end{align}
	where the first equality follows from that the two events $T_1-S_{N(T_1)}$ and $T_2$ are independent, and the second equality follows from Lemma \ref{1lemma:N(T1)} and Lemma \ref{1lemma:bounded}.

	\tcb{As illustrated in~Fig.~\ref{fig:bu-er}},
	\begin{align*}
	\lim_{T\rightarrow\infty}\Eb\left[\frac{R(T)}{T}\right]\leq \frac{\sum_{i=1}^{N(T_1) }X_i^2 + (T_1+T_2-S_{N(T_1)})^2 }{2\Eb[T_1+T_2]}.
	\end{align*}
	
	Consider the channel state realization at the scheduled status updating epochs under BU (and BU-ER) updating. Let $Y_i$ be the duration between the $i$th and $i-1$st epochs when the channel states are good and the corresponding update would be successful if it were sent. Then, $\{Y_i\}_{i=1}^{N(T_1)}$ is identical to $\{X_i\}_{i=1}^{N(T_1)}$. This is because there is no battery outage over $[0, T_1]$, and whether an update is successful or not only depends on the channel state.
	Combining with (\ref{eqn:usebounded-2}), we have
	\begin{align}
	&\lim_{T_0\rightarrow\infty}\lim_{T\rightarrow\infty}\Eb\left[\frac{R(T)}{T}\right]\leq \lim_{T_0\rightarrow\infty} \frac{\Eb[\sum_{i=1}^{N(T_1)}X_i^2 ]}{2\Eb[T_1+T_2]}  \label{eqn:upbound-01} \\
	&\leq \lim_{T_0\rightarrow\infty} \frac{\Eb\left[\sum_{i=1}^{N(T_1)+1}Y_i^2 \right]}{2\Eb\left[\sum_{i=1}^{N(T_1)+1}Y_i-(\sum_{i=1}^{N(T_1)+1}Y_i-T_1)\right]}  \\
	&=\lim_{T_0\rightarrow\infty} \frac{\Eb[N(T_1)+1]\Eb[Y_1^2]}{2\Eb[N(T_1)+1]\Eb[Y_1]-2\Eb\left[\sum_{i=1}^{N(T_1)+1}Y_i-T_1\right]}, \label{eqn:upbound-1} 
	\end{align}

	where (\ref{eqn:upbound-1}) follows from Wald's equality and the fact that $N(T_1)+1$ is a stopping time for $\{Y_i\}$ for any given $T_1$.
	
	Since $\Eb[N(T_1)+1]\Eb[Y_1]\geq \Eb[T_1]$, according to Lemma~\ref{1lemma:N(T1)}, 
	\begin{align}
	\lim_{T_0\rightarrow\infty}\Eb[N(T_1)+1]\Eb[Y_1]\geq \lim_{T_0\rightarrow\infty}\Eb[T_1]=\infty.
	\end{align}
	Meanwhile, we have $\Eb\left[\sum_{i=1}^{N(T_1)+1}Y_i-T_1\right]$ uniformly bounded for any $T_1$
	based on Proposition 3.4.6 in \cite{ross:1996}. Therefore, (\ref{eqn:upbound-1}) is equal to $\frac{\Eb[Y_1^2]}{2\Eb[Y_1]}$, i.e., $\frac{2-p}{2p}$.
\end{proof}

Theorem~\ref{1thm:lower}, Lemma \ref{1lemma:suboptimal} and Lemma~\ref{1lemma:AoI} imply the optimality of the BU updating, as summarized in the following theorem.

\begin{Theorem}[Optimality of BU Updating]  \label{1thm:nofeedback}
Among all policies in $\Pi_3$, the BU updating policy is optimal when updating feedback is unavailable, i.e.,
\begin{align*}
\limsup_{T\rightarrow \infty}\Eb\left[\frac{R(T)}{T}\right] & =  \frac{2-p}{2p}.
\end{align*}
\end{Theorem}

\section{Status Updating With Perfect Feedback} \label{sec:feedback}
In this section, we consider the case where there exists perfect updating feedback to the sensor. 
With perfect updating feedback, the sensor has the choice to retransmit the update immediately or wait and update later, thus leading to optimal solutions different from the no feedback case.
In order to facilitate the analysis, in the following, we focus on another class of online policies, termed as uniformly bounded policies.

\subsection{A Lower Bound}
Define $K_i$ as the number of {\it attempted} updates (including the last successful one) between two successful updates at time $S_{i-1}$ and $S_{i}$ under any online policy in $\Pi_2$. Then, $K_i$ could be any integer number greater than or equal to one. 

\begin{Definition}[Uniformly bounded policy]\label{2dfn:uniform}
Under a policy $\pi\in\Pi_2$, if: 1) there exists a function $g(k)$ such that when $K_i= k$, $X_i\leq g(k)$, $\forall i$, and $\tcb{\Eb[g^2(K_i)]}<\infty$, and 2) $\Eb[M(t)-M(t-\Delta)]\leq C\Delta$ for any $\Delta>0,t>0$, then, $\pi$ is called a uniformly bounded policy.
\end{Definition}
Roughly speaking, the first condition ensures that the source updates frequently so that the AoI at the destination does not grow unbounded in expectation; The second condition requires that the source does not update too frequently in any period of time. Such conditions are consistent with our intuition that the optimal policies should try to maintain a constant $X_i$ as much as possible. We note that uniformly bounded policies do not have to be renewal or Markovian in general. Denote the set of uniformly bounded policies as $\Pi_4$, then $\Pi_4\subset\Pi_2$. 
We have the following lemma.
\begin{Lemma}\label{2lemma:extra-2}
For any $\pi\in\Pi_4$, it must have $\lim_{T\rightarrow\infty}\frac{\Eb \left[X^2_{N(T)+1}\right]}{T}=0$ and $\lim_{T\rightarrow\infty}\frac{\Eb \left[X_{N(T)+1}\right]}{T}=0$.
\end{Lemma}
The proof of this lemma is adapted from the proof of Theorem~3 in \cite{Yang:2017:AoI}, and provided in Appendix~\ref{app:{2lemma:extra}}.

Besides, we also have the following observation.

\begin{Lemma}\label{2lemma:success}
	Under any policy $\pi\in\Pi_4$, it must have $\lim_{T\rightarrow \infty} \frac{\Eb[N(T)]}{T}\leq p .$
\end{Lemma}
\begin{proof}
	First, we observe that
	\begin{align}
	&\lim_{T\rightarrow \infty} \frac{\Eb[\sum_{i=1}^{N(T)+1}K_i]}{T}\leq \lim_{T\rightarrow \infty} \frac{E_0+\Eb[\sum_{i=1}^{N(T)+1}A_i]}{T}\label{eqn:energy_cc}
	\end{align}
	due to the energy causality constraint. We note that $A(t)-t$ is a continuous-time martingale, where $A(t)$ is a Poisson process with parameter one. Therefore, according to the optimal stopping time theorem~\cite{ross:1996}, for any stopping time $\tau$, we have $\Eb[A(\tau)-\tau]=\Eb[A(0)-0]=0$,	i.e., $\Eb[A(\tau)]=\Eb[\tau]$. Since $S_{N(T)+1}$ is a stopping time associated with the past energy arrivals and \tcb{update erasure patterns} under any $\pi\in\Pi_4$, we have $\Eb[A(S_{N(T)+1})]=\Eb[S_{N(T)+1}]$. Plugging it into (\ref{eqn:energy_cc}), we have 
	\begin{align}
	&\lim_{T\rightarrow \infty} \frac{\Eb[\sum_{i=1}^{N(T)+1}K_i]}{T}\leq \lim_{T\rightarrow \infty} \frac{\Eb[S_{N(T)+1}]}{T}
	 \\
	&=1+\lim_{T\rightarrow \infty} \frac{\Eb[X_{N(T)+1}]}{T}=1,\label{eqn:N(T)}
	\end{align}
	where the last equality follows from Lemma~\ref{2lemma:extra-2}.
	
	Besides, we note that under any online policy $\pi\in\Pi_4$, $K_i$ is an i.i.d. geometric random variable with parameter $p$. Therefore, applying Wald's equality, we have 
	\begin{align}
	&\lim_{T\rightarrow \infty} \frac{\Eb[\sum_{i=1}^{N(T)+1}K_i]}{T}=\lim_{T\rightarrow \infty} \frac{\Eb[N(T)+1]\Eb[K_i]}{T}\\
	&=\lim_{T\rightarrow \infty} \frac{\Eb[N(T)+1]}{Tp}.
	\end{align}
	Combining with (\ref{eqn:N(T)}), we have $\lim_{T\rightarrow \infty} \frac{\Eb[N(T)+1]}{T}=\lim_{T\rightarrow \infty} \frac{\Eb[N(T)]}{T}\leq p$.
\end{proof}

In order to obtain a lower bound on the AoI for all $\pi\in\Pi_4$, we will first drop the energy causality constraint, and focus on those online policies that satisfy Lemma~\ref{2lemma:success} and are also uniformly bounded. Denote the set of such policies as $\Pi_5$. Then, we have $\Pi_4\subset\Pi_5$. Since not all policies in $\Pi_5$ would be feasible if the energy causality constraint is imposed, the minimum expected long-term AoI achieved by policies in $\Pi_5$ serves as a {\it lower bound} for policies in $\Pi_4$.

\begin{Theorem}\label{2thm:comparison}
Any policy $\pi\in\Pi_5$ is suboptimal to a renewal policy, i.e., a policy under which the successful updating points $\{S_i\}_{i=1}^{\infty}$ form a renewal process. Besides, under the renewal policy, $X_i$ only depends on $K_i$.
\end{Theorem}
A sketch of the proof is as follows: For any given policy $\pi\in\Pi_5$, we construct a renewal policy based on all possible sample paths under $\pi$. Specifically, our approach is to first average $X_i$ over sample paths with the same $K_i$, so that all factors other than $K_i$ that may affect $X_i$ can be averaged out. Then, we form a \tcb{linear combination} of $X_i$, and use it as the inter-update delay under the new policy. Such a policy is a renewal policy, and each renewal interval only depends on $K_i$. Through rigorous stochastic analysis, we prove that the constructed renewal policy always outperforms the original policy. The detailed proof of Theorem~\ref{2thm:comparison} is provided in Appendix \ref{app:2thm:comparison}.

In the following, we will focus on renewal policies in $\Pi_2$, and identify the AoI-optimal renewal policy.
\begin{Theorem}\label{2thm:renewal}
Under the optimal renewal policy in $\Pi_5$, $X_i$ equals a constant $\frac{1}{p}$ irrespective of $K_i$, and the corresponding long-term average AoI equals $\frac{1}{2p}$.
\end{Theorem}

\begin{Proof}
\tcb{Based on proof of Theorem~3, under the optimal renewal policy, $X_i$ can only take values from a countable set of constants $\{x_1,x_2,\ldots\}$, depending on the realization of $K_i$. Specifically, $X_i$ will equal $x_k$ if $K_i=k$. Note that $K_i$ is a geometric random variable with parameter $p$ irrespective of the values of $x_k$s. Then, to minimize the expected long-term average AoI, it suffices to solve the following optimization problem:
\begin{align} \label{2eqn:renewalsolve}
\min_{\{x_k\}} \frac{\Eb[ X_i^2]}{2\Eb[X_i]}\quad\mbox{s.t.}\quad \frac{1}{\Eb[X_i]}\leq p,
\end{align}
where the constraint follows from Lemma~\ref{2lemma:success} and the property of renewal processes.}

\tcb{Applying the inequality that $\Eb[X^2]\geq \Eb^2[X]$ to the objective function and utilizing the constraint $\frac{1}{\Eb[X_i]}\leq p$, we have
 \begin{align}
 \frac{\Eb[ X_i^2]}{2\Eb[X_i]}\geq \frac{\Eb[X_i]}{2}\geq \frac{1}{2p},
 \end{align}
 where the equalities can be met if $X_i=\Eb[X_i]=\frac{1}{p}$.}
 
 \if{0}

This is a non-linear fractional programming problem and can be solved using the parametric approach in~\cite{Dinkelbach:1967}. Below we provide an alternative yet simpler approach.

We note that
\begin{align}
\frac{\sum_{k=1}^\infty x^2_k p_k}{2\sum_{k=1}^\infty x_k p_k} 
\geq \frac{(\sum_{k=1}^\infty x_k p_k)^2}{2\sum_{k=1}^\infty x_k p_k} 
=\frac{1}{2}\sum_{k=1}^\infty x_k p_k  
\geq \frac{1}{2p}  \label{2eqn:renewalsolve-2},
\end{align}
where the first inequality in (\ref{2eqn:renewalsolve-2}) follows from Jensen's inequality and the second one follows from the constraint in (\ref{2eqn:renewalsolve}). The equalities hold if and only if $x_k=x_1$ for all $k=1,2,\ldots$ and $\sum_{k=1}^\infty x_k p_k = \frac{1}{p}$. Combining with $p_k=(1-p)^{k-1}p$, we have $x_k=\frac{1}{p}$ for all $k=1,2,\ldots$.
Thus, the solution of (\ref{2eqn:renewalsolve}) is $x_k=\frac{1}{p}$, for all $k=1,2,\ldots$ and  the corresponding minimum is $\frac{1}{2p}$.
\fi
\end{Proof}

Combining Theorem~\ref{2thm:comparison} and Theorem~\ref{2thm:renewal}, we obtain a lower bound for all $\pi\in\Pi_4$ as follows.

\begin{Theorem}\label{2thm:lower}
{\bf (Lower Bound for Channel with Perfect Feedback)} For any policy $\pi\in\Pi_4$, the expected long-term average AoI is lower bounded by $\frac{1}{2p}$.
\end{Theorem}

\subsection{Optimal Online Status Updating}
Motivated by the uniform structure of $\{X_i\}$ under the optimal renewal policy in Theorem \ref{2thm:renewal}, we define the {\it Best-effort Uniform updating with Retransmission} (BUR) policy as follows.

\begin{Definition}[BUR Updating]
The sensor is scheduled to update the status at $s_n=n/p$, $n=1,2,\ldots$. The sensor keeps sending updates at $s_n$ until an update is successful or until it runs out of battery; Otherwise, the sensor keeps silent until the next scheduled status update time.
\end{Definition}

In order to prove that the BUR updating policy is optimal, we will first construct a sequence of policies which are sub-optimal to the BUR updating policy, and show that the limit of those suboptimal policies achieves the lower bound in Theorem~\ref{2thm:lower}.

\begin{Definition}[BUR with Energy Removal (BUR-ER$_{T_0}$)]
	The sensor performs BUR updating policy until the battery level after sending an update becomes zero for the first time, or until time $T_0^+$, in which case the sensor depletes its battery after a successful update at $T_0$; After that, when the battery level becomes higher than or equal to one after a successful update for the first time, the sensor reduces the battery level to one, and then repeats the process.
\end{Definition}

\begin{Lemma} \label{2lemma:suboptimal}
	The BUR-ER$_{T_0}$ updating policy is suboptimal to the BUR updating policy.
\end{Lemma} 
\begin{proof}
	We note that the BUR-ER$_{T_0}$ updating policy is identical to the BUR updating policy up to the energy removal step. Given the same energy harvesting sample path, the battery level under BUR is always higher than that under BUR-ER$_{T_0}$. Thus, BUR-ER$_{T_0}$ incurs more infeasible status updating points. With the same \tcb{update erasure pattern}, the instantaneous AoI under BUR-ER$_{T_0}$ is always greater than or equal to that under BUR sample path-wisely. Thus, the expected time-average AoI under BUR-ER$_{T_0}$ is greater than or equal to that under BUR. 
\end{proof}

Note that BUR-ER$_{T_0}$ updating is a renewal policy and Fig.~\ref{fig:bur-er} is an illustration of one renewal interval. In order to analyze the expected long-term average AoI, it suffices to analyze the expected average AoI over one renewal interval. Thus, we will focus on the first renewal interval, and show that the expected average AoI converges to the lower bound in Theorem~\ref{2thm:lower}.
The renewal interval consists of two stages. The first stage starts at time zero and ends until the battery becomes empty for the first time, or until time $T_0^+$, denoted as $T_1$. We note that all scheduled updating points over $(0, T_1)$ are feasible. The second stage starts at $T_1$ and ends when the battery level after a successful update becomes higher than or equal to one for the first time after $T_1$, denoted as \tcb{$T_1+T_2$, where $T_2$ is the duration of the second stage}.

\begin{figure}[t]\centering
	\includegraphics[width=3.3in]{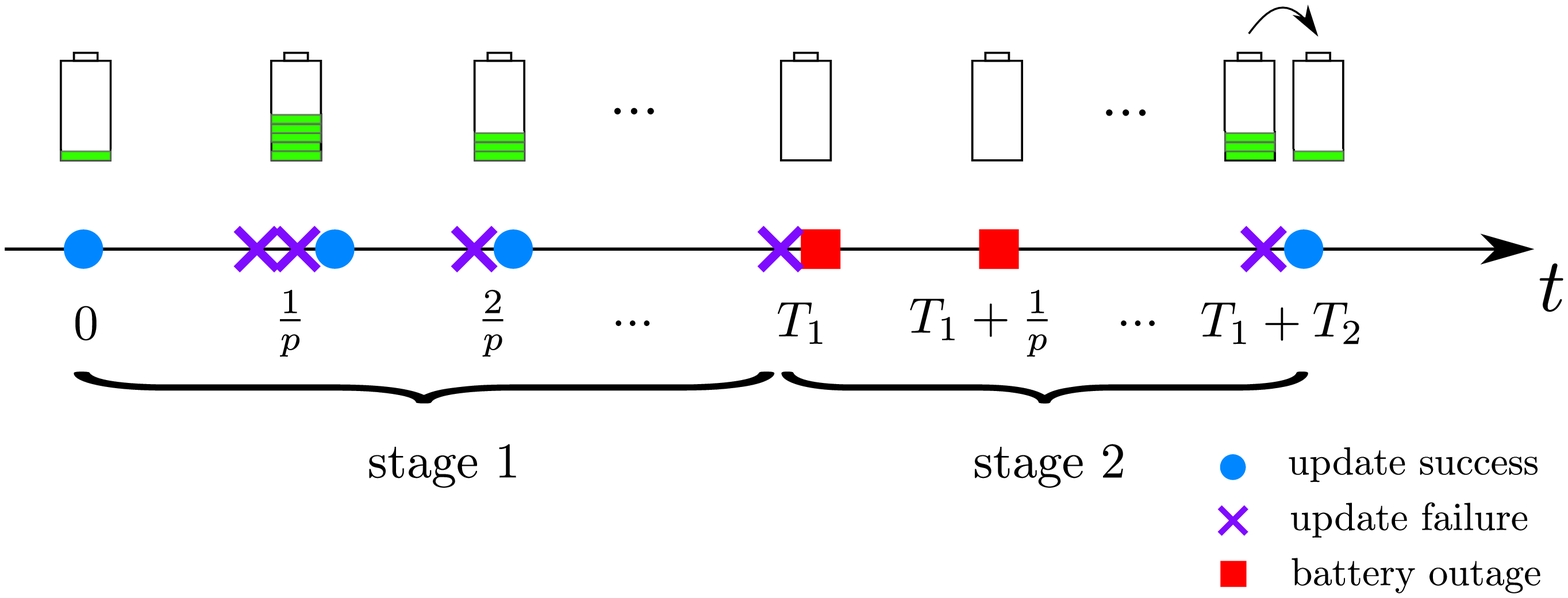}
	\vspace{-0.05in}
	\caption{\tcb{An illustration of the BUR-ER$_{T_0}$ updating policy and the battery level right after each updating epoch.  AoI will be reset to zero at the successul updating epochs.} }\label{fig:bur-er}
	\vspace{-0.1in}
\end{figure}

\if{0}
\begin{figure}[t]\centering
	\includegraphics[width=3.3in]{BURER_revision.eps}
	\vspace{-0.05in}
	\caption{\tcb{An illustration of the BUR-ER$_{T_0}$ updating policy and the battery level right after each updating epoch.  AoI will be reset to zero at the successul updating epochs.} }\label{fig:bur-er}
	\vspace{-0.1in}
\end{figure}
\fi

\begin{Lemma} \label{2lemma:N(T1)}
	Under BUR-ER$_{T_0}$ updating, $\lim_{T_0\rightarrow \infty} \Eb[T_1]=+\infty.$
\end{Lemma}
\begin{proof}
	Consider a ``random walk" $\{\Omega_n\}_{n=0}^{\infty}$. It starts with $1$ and the evolves as $\Omega_n=(\Omega_{n-1}+A_{n}-B_{n})^+$, where $A_n$ is an i.i.d. Poisson random variable with parameter $\frac{1}{p}$ and $B_n$ is an i.i.d. geometric random variable with parameter $p$. Denote the first zero-hitting time for $\{\Omega_n\}_{n=0}^{\infty}$ as $T_1$. Then $\Omega_0=1$ and $\Omega_{T_1} = 0$. We note that when $T_0=\infty$, $\{\Omega_n\}_{n=0}^{T_1}$ is identical to the battery level evolution process $\{E(s_n^+)\}_{n=0}^{T_1}$ under the BUR-ER$_{T_0}$ updating policy.
	
	For ease of exposition, define $C_n:=A_n-B_n$, and $\gamma(\alpha):=\log\Eb[e^{-\alpha C_n}]$ for $\alpha>0$. Then,
	we have 
	\begin{align}\label{2eqn:gamma}
	\Eb[e^{-\alpha C_n-\gamma(\alpha)}]=1.
	\end{align} 
	Based on the definition of $A_n$, $B_n$ and $C_n$, we have 
	\begin{align}
	&\Eb[e^{-\alpha C_n}]
	=e^{\frac{1}{p}(e^{-\alpha}-1)}\frac{pe^\alpha}{1-(1-p)e^\alpha}.
	\end{align}
	Therefore,
	\begin{align}
	\gamma(\alpha)&=\log \Eb[e^{-\alpha C_n}]	=\frac{1}{p}(e^{-\alpha}-1)+\log \frac{pe^\alpha}{1-(1-p)e^\alpha}. \label{2eqn:marting-m1}
	\end{align}
	Taking derivative of (\ref{2eqn:marting-m1}), we get
	\begin{align}
	&\gamma'(\alpha)=-\frac{1}{p}e^{-\alpha}+\frac{1}{1-(1-p)e^\alpha}. \label{2eqn:marting-5}
	\end{align}

	Next, we define a process associated with $\{\Omega_n\}_{n=0}^{\infty}$ as $\{e^{-\alpha \Omega_n - n\gamma(\alpha)}\}_{n=0}^{\infty}$.
	We note that
	\begin{align}
	&\Eb [ e^{-\alpha \Omega_k -\gamma(\alpha)k} | \Omega_1,\ldots, \Omega_{k-1} ]  \nonumber\\
	&=\Eb [ e^{-\alpha (\Omega_{k-1}+C_{k})^+ -\gamma(\alpha)k} | \Omega_1,\ldots, \Omega_{k-1} ]  \nonumber\\
	&\leq  \Eb [ e^{-\alpha (\Omega_{k-1}+C_{k}) -\gamma(\alpha)k} | \Omega_1,\ldots, \Omega_{k-1} ] \nonumber\\
	&=e^{-\alpha \Omega_{k-1}-\gamma(\alpha)(k-1)} \Eb[e^{-\alpha C_{k}-\gamma(\alpha)}]   \nonumber\\
	&= e^{-\alpha \Omega_{k-1}-\gamma(\alpha)(k-1)} ,\label{2eqn:marting-1}
	\end{align}
	where (\ref{2eqn:marting-1}) follows from (\ref{2eqn:gamma}). Therefore, $\{e^{-\alpha \Omega_n - n\gamma(\alpha)}\}_{n=0}^{\infty}$ is a super-martingale process, i.e.,
	\begin{align*}  
	e^{-\alpha \Omega_0} &\geq \Eb [e^{-\alpha \Omega_{T_1} - \gamma(\alpha)T_1}] \geq \Eb[1-(\alpha \Omega_{T_1} +T_1 \gamma(\alpha))].
	\end{align*}
	Since $\Omega_0=1$ and $\Omega_{T_1} = 0$, combining with (\ref{2eqn:marting-5}), we have
	\begin{align}
	\Eb[T_1] &\geq \lim_{\alpha \rightarrow 0^+} \frac{1-e^{-\alpha \Omega_0} }{\gamma(\alpha)} = \lim_{\alpha \rightarrow 0^+}\frac{\Omega_0e^{-\alpha\Omega_0}}{\gamma'(\alpha)} = \infty.\label{2eqn:marting-4}
	\end{align}	
\end{proof}

\begin{Lemma}\label{2lemma:bounded}
	Under the BUR-ER$_{T_0}$ updating policy, $\Eb[T_2]$, $\Eb[T^2_2]$ are uniformly bounded.
\end{Lemma}
\begin{proof}
	Under BUR-ER$_{T_0}$ updating policy, the number of energy arrivals over $[\frac{n}{p},\frac{n+1}{p})$ (denoted as $A_{n+1}$) is a Poisson random variable with parameter $1/p$. If the source has sufficient energy, the total number of attempts at time $\frac{n+1}{p}$ (denoted as $B_{n+1}$) is an  i.i.d. geometric random variable with parameter $p$. Therefore, if the battery is empty at time $\frac{n}{p}$, it will increase to one or above after a successful update at time $\frac{n+1}{p}$ only when $A_{n+1}-B_{n+1}\geq 1$, which will happen with a constant probability. Thus, \tcb{$pT_2$} is a geometric random variable whose first and second moments are finite.
\end{proof}

\begin{Lemma} \label{2lemma:upperbound}
	As $T_0\rightarrow\infty$, the expected long-term average AoI under BUR-ER$_{T_0}$ updating is upper bounded by $\frac{1}{2p}$.
\end{Lemma}
\begin{proof}
	First, we note that 
	\begin{align}
	\lim_{T_0\rightarrow \infty}\frac{\Eb[ (T_1+T_2-S_{N(T_1)})^2 ]}{2\Eb[T_1+T_2]}\leq \lim_{T_0\rightarrow \infty}\frac{\Eb[ (T_2+\frac{1}{p})^2 ]}{2\Eb[T_1]} =0 , \label{2eqn:usebounded-11} 
	\end{align}
	where (\ref{2eqn:usebounded-11}) follows from the fact that $T_1-S_{N(T_1)}$ is upper bounded by $1/p$ under the BU-ER$_{T_0}$ policy, Lemma \ref{2lemma:N(T1)} and Lemma \ref{2lemma:bounded}.
	
	Next, we note that the BU-ER$_{T_0}$ updating policy is a renewal policy and the expected long-term average AoI is equal to the expected average AoI over one renewal interval. \tcb{Therefore,}
	\begin{align}
	&\lim_{T_0\rightarrow\infty}\lim_{T\rightarrow\infty}\Eb\left[\frac{R(T)}{T}\right]\nonumber\\
	&	\leq \lim_{T_0\rightarrow\infty}\frac{\Eb[\sum_{i=1}^{N(T_1) }X_i^2 + (T_1+T_2-S_{N(T_1)})^2 ]}{2\Eb[T_1+T_2]} \\
	&\leq \lim_{T_0\rightarrow\infty}\frac{\Eb[\sum_{i=1}^{N(T_1)} X_i^2 ]}{2\Eb[S_{N(T_1)}]}  =\lim_{T_0\rightarrow\infty} \frac{\Eb[N(T_1)] \frac{1}{p^2}}{2 \Eb[N(T_1)] \frac{1}{p}} =\frac{1}{2p}   ,\label{2eqn:usebound-3}
	\end{align}
	where (\ref{2eqn:usebound-3}) follows from (\ref{2eqn:usebounded-11}) and the fact that $X_i=1/p$ for $i \leq N(T_1)$ and $S_{N(T_1)}=N(T_1)/p$.
\end{proof}

Lemma \ref{2lemma:upperbound} indicates that the expected time-average AoI under the BUR-ER$_{T_0}$ updating policy converges to the lower bound in Theorem~\ref{2thm:lower} as $T_0$ goes to infinity. According to Lemma \ref{2lemma:suboptimal}, BUR-ER$_{T_0}$ is suboptimal to BUR. Therefore, the BUR updating policy also achieves the lower bound, thus it is optimal.
We summarize the optimality result in the next theorem.

\begin{Theorem}[Optimality of BUR Updating]\label{2thm:myopic}
Among all policies in $\Pi_4$, the BUR updating policy is optimal when transmission feedback is available, i.e.,
\begin{align*}
\limsup_{T\rightarrow \infty}\Eb \left[\frac{R(T)}{T} \right] & =  \frac{1}{2p} .
\end{align*}
\end{Theorem}

\section{Simulation Results} \label{sec:simulation}
In this section, we evaluate the performances for the proposed status updating policies through simulations. For each case, we generate sample paths for the Poisson energy harvesting process with $\lambda=1$ and compute the sample average of the time average AoI over $1000$ sample paths.

\subsection{Status Updating Without Feedback}
First, we evaluate the BU updating policy in Fig.~\ref{fig:simu-1}. We vary $p=0.2,0.6,1.0$, and plot both the time average AoI as a function of $T$ and the corresponding lower bound in the figure. We observe that all time average AoI curves gradually approach the corresponding lower bound $\frac{2-p}{2p}$ as $T\rightarrow\infty$. The results show that the proposed BU updating policy is optimal. Note that the time average AoI is monotonically decreasing as $p$ increases. This is intuitive since channel with better quality, i.e., larger $p$, will render smaller time average AoI.

\begin{figure}[t]
	\centering
	\includegraphics[width=3.5in]{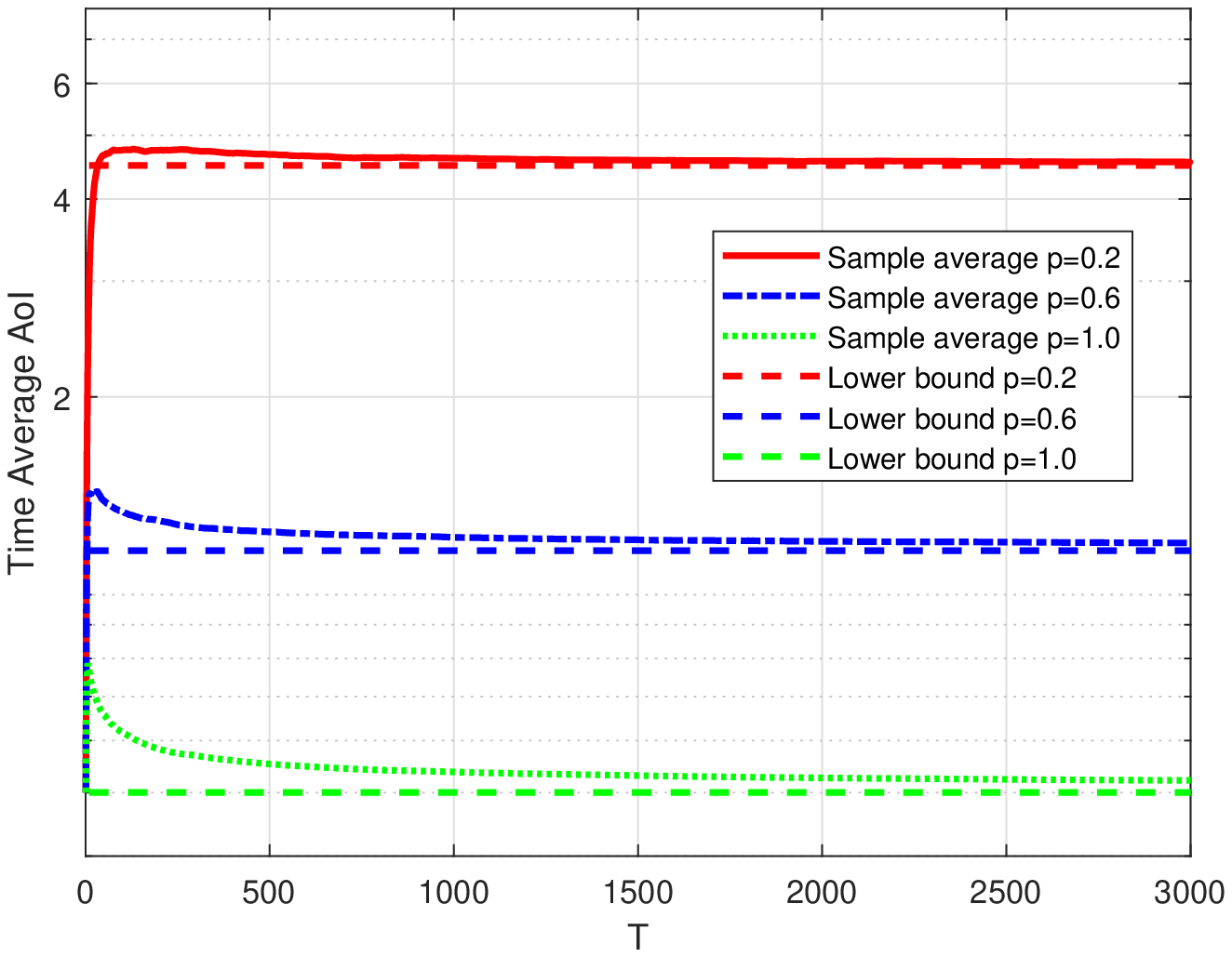}
	\caption{Performances of BU policy.}
	\label{fig:simu-1}
\end{figure}
\begin{figure}[t]
	\centering
	\includegraphics[width=3.5in]{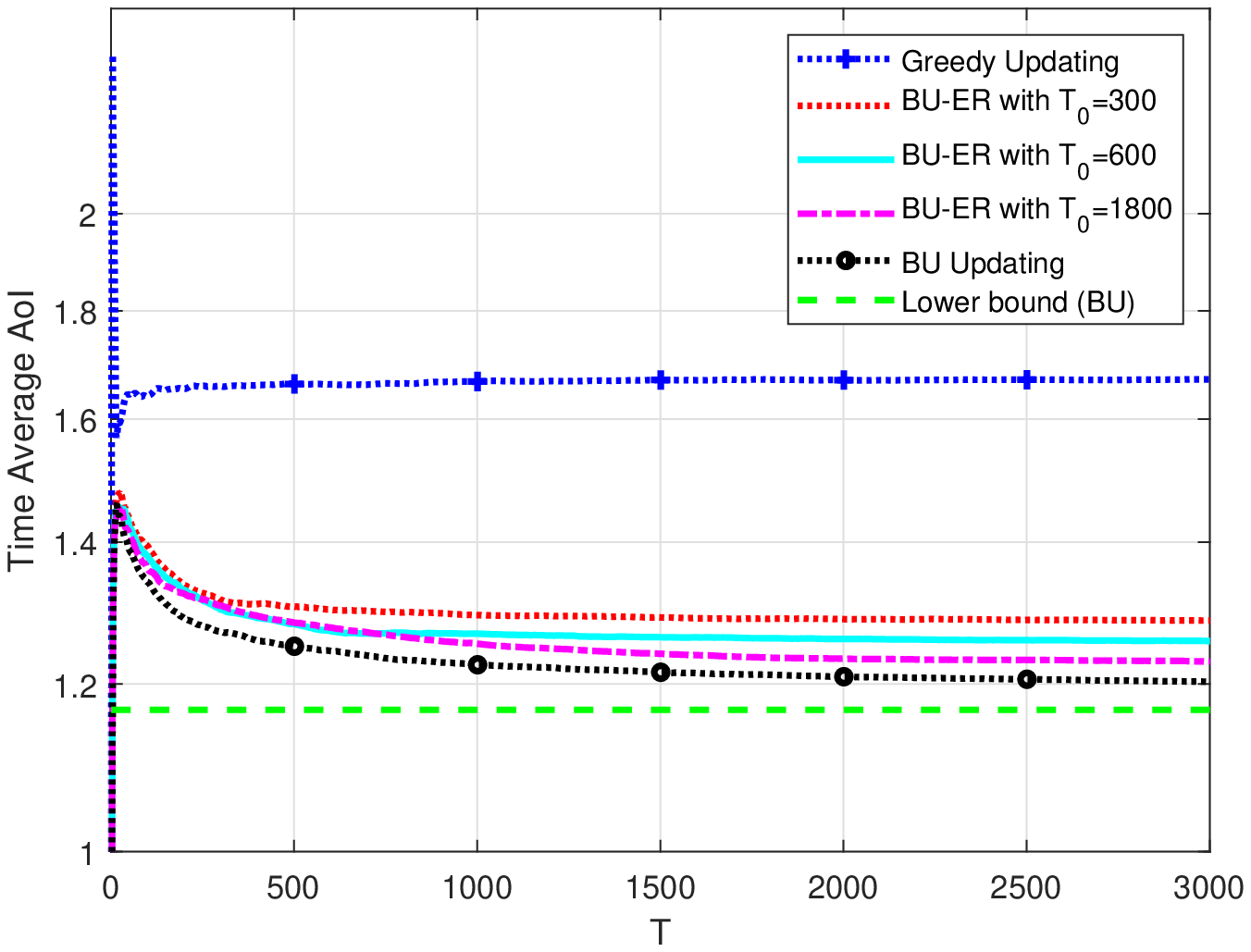}
	\caption{Performances of BU-ER policy.}
	\label{fig:simu-2}
\end{figure}
\begin{figure}[t]
	\centering
	\includegraphics[width=3.5in]{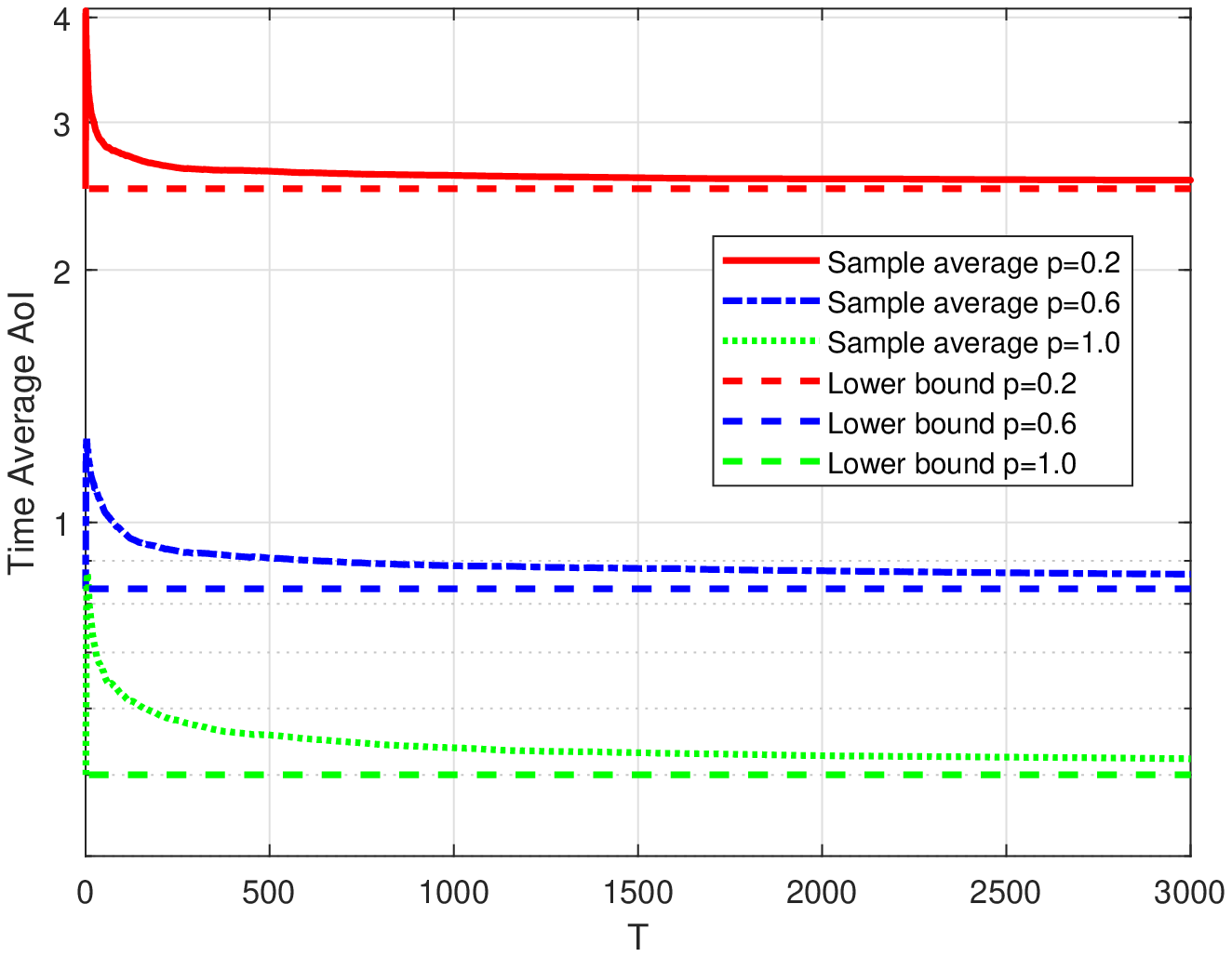}
	\caption{Performances of BUR policy.}
	\label{fig:simu-3}
\end{figure}
\begin{figure}[t]
	\centering
	\includegraphics[width=3.5in]{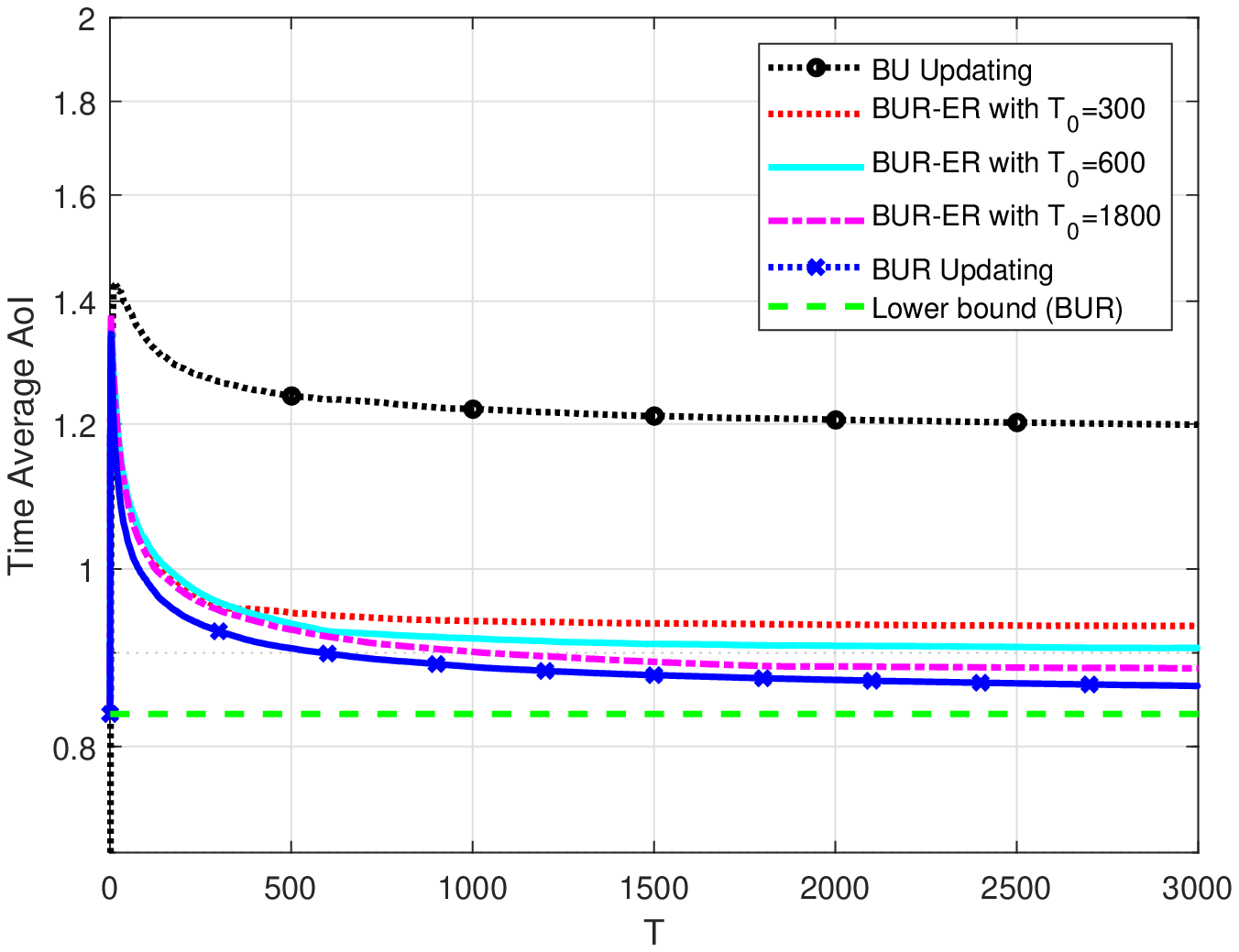}
	\caption{Performances of BUR-ER policy.}
	\label{fig:simu-4}
\end{figure}

Next, we evaluate the performances of virtual policies BU-ER$_{T_0}$ for different value of $T_0$ in Fig.~\ref{fig:simu-2}. We fix $p=0.6$ and plot the time average AoI under BU-ER$_{T_0}$ with $T_0=300,600,1800$. We also compare with a greedy updating policy and the BU updating policy. Under the greedy updating policy, the sensor updates instantly when one unit of energy arrives. As we observe in Fig.~\ref{fig:simu-2}, the greedy policy results in the highest average AoI, and never approaches the lower bound. The time averaged AoI under the BU-ER$_{T_0}$ updating policy is monotonically decreasing as $T_0$ increases, and gradually approaches that under the BU updating policy. This is consistent with Lemma~\ref{1lemma:suboptimal} and Lemma~\ref{1lemma:AoI} that BU-ER$_{T_0}$ updating is sub-optimal to BU updating, and eventually converges to it when $T_0$ increases.

\subsection{Status Updating With Perfect Feedback}

\if{0}
\begin{figure}[t]
	\centering
	\includegraphics[width=3.4in]{simu_nofb_1.eps}
	\caption{Performances of BU policy.}
	\label{fig:simu-1}
\end{figure}
\begin{figure}[t]
	\centering
	\includegraphics[width=3.4in]{nofeedback_greedy_compare-v2.eps}
	\caption{Performances of BU-ER policy.}
	\label{fig:simu-2}
\end{figure}

\subsection{Status Updating With Perfect Feedback}
\begin{figure}[t]
	\centering
	\includegraphics[width=3.4in]{simu_fb_1.eps}
	\caption{Performances of BUR policy.}
	\label{fig:simu-3}
\end{figure}
\begin{figure}[t]
	\centering
	\includegraphics[width=3.4in]{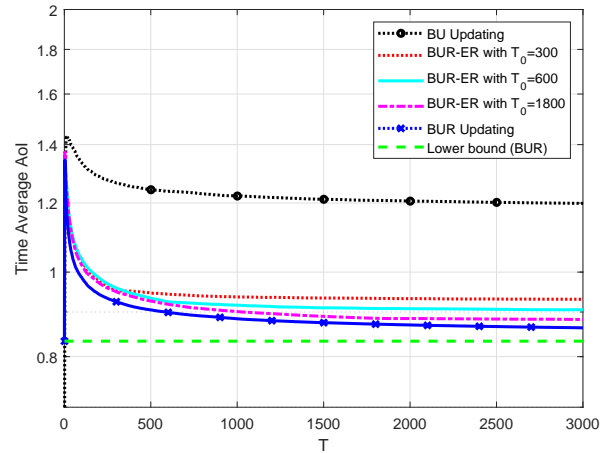}
	\caption{Performances of BUR-ER policy.}
	\label{fig:simu-4}
\end{figure}
\fi
Next, we evaluate the performances of the proposed online policies when perfect feedback is available to the sensor. In Fig.~\ref{fig:simu-3}, under the BUR updating policy, we plot the time average AoI with $p=0.2,0.6,1.0$ and the corresponding lower bound $\frac{1}{2p}$. We note that as $T\rightarrow\infty$, the time average AoI approaches the lower bound. Thus BUR updating is optimal. We then evaluate the performances of the BUR-ER$_{T_0}$ updating policy in Fig.~\ref{fig:simu-4}. We fix $p=0.6$, choose $T_0=300,600,1800$ and plot the time average AoI as a function of $T$. As a comparison, we also plot the time average AoI under the BU updating policy and the BUR updating policy in the figure. We note that the AoI under BUR-ER$_{T_0}$ gradually decreases and approaches that under the BUR updating policy as $T_0$ increases, which is consistent with Lemma~\ref{2lemma:suboptimal} and Lemma~\ref{2lemma:upperbound}. The performance gap between the BU updating and the BUR updating indicates that exploiting updating feedback can significantly reduces time average AoI in the system.

\section{Conclusions} \label{sec:conclusion}
In this paper, we considered the optimal online status update policies for an energy harvesting source in presence of updating erasures. We investigated both cases where no updating feedback or perfect feedback is available to the source. 
For each case, we first obtained a lower bound and then proved the proposed status updating policy can achieve the lower bound among a broadly defined class of policies. The optimality of proposed status update policies were proved through constructing a sequence of virtual status updating policies which are sub-optimal to the original policy and asymptotically achieve the lower bound. The performances of the proposed policies were evaluated through simulations. We point out that although we only showed the optimality of the proposed policies within a subset of online policies, we conjecture that their optimality can be extended for all online policies. How to generalize the results is one of our future steps. Another direction we would like to pursue is to investigate the impact of update erasures on the optimal updating policy for an EH source with finite battery.

\appendix
\subsection{Proof of Theorem \ref{1thm:lower}}  \label{app:1thm:lower}

Define $S_{i}^{T}:=\min\{S_i,T\}$, $l_{n}^{T}:=\min\{l_n,T\}$, and $p_n:=(1-p)^{n-1}p$. Then, under any $\pi\in\Pi_3$, the expected average AoI over $[0,T]$ can be expressed as
\begin{align} 
&\Eb\left[\frac{R(T)}{T}\right]  
=\frac{1}{T}\Eb\left[\sum_{i=0}^{N(T)}\frac{(\tcb{S_{i+1}^T} - S_i)^2}{2}\right] 
\label{1eqn:lowbound-11} \\
&=\frac{1}{2T}\Eb\left[ \sum_{n=1}^{M(T)} p_n l_n^2+\left(1-\sum_{n=1}^{M(T)}p_n\right) T^2 \right. \nonumber \\
&\left. \qquad + \sum_{n=1}^{M(T)} \sum_{j=1}^{\infty} (l_{n+j}^T - l_n)^2p p_j 
\right], \label{1eqn:lowbound-12}
\end{align}
where the first two terms inside the expectation in (\ref{1eqn:lowbound-12}) correspond to the AoI contribution over $[0,S^T_1]$, and the last term correspond to the AoI contribution over any other $[S_i,S^T_{i+1}]$. This can be explained as follows. With fixed updating epochs $\{l_n\}$, depending on the realization of the channel state, the interval $[0,T]$ can be decomposed into segments, separated by successful updates. The probability to have $[l_n, l^T_{n+j}]$, $1\leq n\leq M(T),j\geq 1$, as one of such segment equals $pp_j$, which corresponds to the event that update at $l_n$ succeeds, and the next successful update is at $l_{n+j}$. The corresponding AoI contribution over $[l_n, l^T_{n+j}]$ thus needs to be weighted by $pp_j$ when the expected AoI is calculated. Since the AoI contribution over $[0,S^T_1]$ is always positive, in the following, we will drop it to obtain a lower bound, i.e., 
\begin{align}
&\lim_{T\rightarrow \infty}\Eb\left[\frac{R(T)}{T}\right]  \nonumber \\
&\geq \lim_{T\rightarrow \infty}\frac{1}{2T} \Eb\left[ p\sum_{j=1}^{\infty}p_j \sum_{n=1}^{M(T)} (l_{n+j}^{T}-l_n)^2  \right]\\
&\geq \lim_{T\rightarrow \infty} \frac{1}{2T} \Eb\left[ p\sum_{j=1}^{\infty}p_j 
\frac{1}{M(T)} \left(\sum_{n=1}^{M(T)}(l_{n+j}^T-l_n) \right)^2  \right] \label{1eqn:lowbound-3revise} \\
& = \lim_{T\rightarrow \infty} \frac{1}{2T} \Eb\left[ p\sum_{j=1}^{\infty}p_j \frac{1}{M(T)} \left(jT-\sum_{n=1}^{j} l^T_n\right)^2  \right]\label{1eqn:lowbound-3}  \\
&
= \lim_{T\rightarrow \infty} \frac{1}{2}p\sum_{j=1}^{\infty}p_{j} j^2 \Eb\left[\frac{(T-\bar{l}_j^T)^2}{M(T)T} \right],\label{1eqn:lowbound-31}
\end{align}
where \tcb{(\ref{1eqn:lowbound-3revise}) is based on a consequence of Jensen's inequality that $\frac{1}{n}\sum_{i=1}^n x_i^2\geq \left(\frac{1}{n}\sum_{i=1}^n x_i\right)^2$ for any $x_i\in\Rb$ and (\ref{1eqn:lowbound-3}) is obtained after rearranging the items in the  summation in (\ref{1eqn:lowbound-3revise}) and considering the cases $j\leq M(T)$ and $j>M(T)$ separately. After extracting a factor $j^2$ from the squared summation in (\ref{1eqn:lowbound-3}) and pushing the factor $\frac{1}{T}$ and the expectation operator into the summation, we obtain (\ref{1eqn:lowbound-31}), where $\bar{l}_j^T:=\sum_{n=1}^{j}l_n^T/j$.}

Since each term in the summation in (\ref{1eqn:lowbound-31}) is positive, we can switch the order of limit and summation. We note that for any given $j$, $\Eb[\bar{l}_j^T]\leq \Eb[l_j]<\infty$ according to the definition of bounded policy. Besides, for any policy that renders a finite expected average AoI, we must have $\lim_{T\rightarrow \infty}M(T)=\infty$ almost surely according to Lemma~\ref{1lemma:finiteAoI}. Therefore, according to the bounded convergence theorem~\cite{Royden:Math}, we have
\begin{align}
\lim_{T\rightarrow \infty}\Eb\left[\frac{\bar{l}^T_j}{M(T)}\right]=0, \quad \lim_{T\rightarrow \infty}\Eb \left[\frac{(\bar{l}_j^T)^2}{M(T)T} \right]=0 .
\end{align} 
Combining with (\ref{1eqn:lowbound-31}), we have
\begin{align}
&\lim_{T\rightarrow \infty}\Eb\left[\frac{R(T)}{T}\right]\geq 
\frac{1}{2}p\sum_{j=1}^{\infty}p_j j^2 \lim_{T\rightarrow \infty} \Eb\left[\frac{T}{M(T)}\right]
\\
&\geq\frac{1}{2}p \sum_{j=1}^{\infty} j^2 (1-p)^{j-1}p  =\frac{2-p}{2p}, \label{1eqn:lowbound-4}
\end{align}
where the first inequality follows from Lemma~\ref{1lemma:rate}.

\subsection{Proof of Lemma~\ref{2lemma:extra-2}}  \label{app:{2lemma:extra}}
	We first prove $\lim_{T\rightarrow\infty}\frac{\Eb \left[X^2_{N(T)+1}\right]}{T}=0$.
	
	Denote $F_n(t)$ as the cumulative distribution function of $S_n$ under a uniform bounded policy, i.e., $F_n(t)=\Pb[S_n\leq t]$.
	Recall that $N(t)$ is the number of status updates successfully received at the destination over $(0,t]$. We have
	\begin{align}
	\Eb[N(t)]=\sum_{n=0}^{\infty} F_n(t)  \label{2eqn:extra-0}  .
	\end{align}

We note that
	\begin{align}
	&\Eb[X_{n+1}^2 \lv_{S_{n+1}>T} | S_n=t ]  \nonumber \\
	&=\Eb[X_{n+1}^2 \lv_{X_{n+1}>T-t} | S_n=t ]  \\
	&\leq \Eb_k [ g^2(k) \lv_{g(k)>T-t} | S_n=t, K_{n+1}=k  ] \label{2eqn:extra-1} \\
	&=\Eb_k [ g^2(k) \lv_{g(k)>T-t}  | K_{n+1}=k ] \label{2eqn:extra-2} \\
	&:=G(T-t)  \label{2eqn:extra-12} ,
	\end{align}
	where (\ref{2eqn:extra-1}) follows from the definition of uniformly bounded policy and (\ref{2eqn:extra-2}) follows from the fact that $g(k)$ is independent of other parameters. 
	We note that
	\begin{align}
	\lim_{\Delta \rightarrow \infty} G(\Delta) = 0  \label{2eqn:extra-13} .
	\end{align}

	Besides,
	\begin{align}
	&\Eb[X_{N(T)+1}^2] \nonumber\\
	&=\sum_{n=0}^{\infty} \int_{0}^{T} \Eb[X_{n+1}^2 \lv_{S_{n+1}>T} | S_n=t ] dF_n(t)  \\
	&\leq \int_{0}^{T} G(T-t) d\left( \sum_{n=0}^{\infty}F_n(t) \right)  \label{2eqn:extra-3} \\
	&=\int_{0}^{T} G(T-t) d\Eb[N(t)]  \label{2eqn:extra-4}  ,
	\end{align}
	where (\ref{2eqn:extra-3}) follows from (\ref{2eqn:extra-12}), and (\ref{2eqn:extra-4}) follows from (\ref{2eqn:extra-0}).
	
	For any fixed $\Delta$ satisfying $0\leq \Delta \leq T$, we have
\begin{align}
&\frac{1}{T}\int_{0}^{T} G(T-t) d\Eb[N(t)] \nonumber\\
&=\frac{1}{T} \hspace{-0.04in} \int_{0}^{T-\Delta} \hspace{-0.05in} G(T-t) d\Eb[N(t)]  \hspace{-0.02in} + \hspace{-0.02in} \frac{1}{T}\hspace{-0.04in} \int_{T-\Delta}^{T} \hspace{-0.05in} G(T-t) d\Eb[N(t)]  \\
&\leq G(\Delta)\frac{\Eb[N(T-\Delta)]}{T} + G(0)\frac{\Eb[N(T)]-\Eb[N(T-\Delta)]}{T}   \label{2eqn:extra-5}  ,
\end{align}
	where (\ref{2eqn:extra-5}) follows from that fact that $G(t)$ is a non-increasing function.
	
	Recall that $M(t)$ is defined as the total number of attempted status updates over $(0,t]$, which is upper bounded by the total number of energy arrivals $A(t)+E_0$ due to the energy causality constraint. We observe that
	\begin{align}
	&\lim_{T\rightarrow\infty} G(\Delta)\frac{\Eb[N(T-\Delta)]}{T}  \nonumber\\
	&	=\lim_{T\rightarrow\infty} G(\Delta) \frac{\Eb[pM(T-\Delta)]}{T} \label{2eqn:extra-a1} \\
	&\leq\lim_{T\rightarrow\infty} G(\Delta) \frac{p\Eb[A(T-\Delta)+E_0]}{T} \\
	&	=\lim_{T\rightarrow\infty} G(\Delta) \frac{p(T-\Delta+E_0)}{T} 
	=pG(\Delta)    \label{2eqn:extra-a4}.
	\end{align}

	Based on the definition of uniformly bounded policy in Definition \ref{2dfn:uniform}, we have
	\begin{align}
	&\lim_{T\rightarrow\infty}\frac{\Eb[N(T)]-\Eb[N(T-\Delta)]}{T} \nonumber\\
	&=\lim_{T\rightarrow\infty}\frac{p\Eb[M(T)]-p\Eb[M(T-\Delta)]}{T}  \\
	&\leq \lim_{T\rightarrow\infty} \frac{pC\Delta}{T}  
	=0 \label{2eqn:extra-b1}.
	\end{align}

	Combining (\ref{2eqn:extra-5}), (\ref{2eqn:extra-a4}) and (\ref {2eqn:extra-b1}), we have
	\begin{align}
	\lim_{T\rightarrow \infty} \frac{1}{T}\int_{0}^{T} G(T-t) d\Eb[N(T)]  =pG(\Delta)  \label{2eqn:extra-8} 
	\end{align}
	for any $\Delta \geq 0$. Therefore, by letting $\Delta\rightarrow\infty$ we have $\lim_{T\rightarrow\infty}\frac{\Eb \left[X^2_{N(T)+1}\right]}{T}=\lim_{\Delta \rightarrow \infty}pG(\Delta)=0$, where the last equality follows from (\ref{2eqn:extra-13}).
	
   Since $\Eb^2[X_{N(T)+1}]\leq \Eb[X_{N(T)+1}^2]$, we have $\lim_{T\rightarrow\infty}\frac{\Eb \left[X_{N(T)+1}\right]}{T}=0$ as well.

\subsection{Proof of Theorem \ref{2thm:comparison}}  \label{app:2thm:comparison}
The proof is adapted from the proof of Theorem~3 in \cite{Yang:2017:AoI}. For the completeness of this paper, we provide the detailed proof here.

We define
\begin{align}  
\hat{X}_{i+1}(k) &:=\Eb[X_{i+1}| \tcb{i\leq N(T)}, K_{i+1}=k] \\
&=\frac{\Eb[X_{i+1} \lv_{i\leq N(T)}| K_{i+1}=k]}{\Eb[ \lv_{i\leq N(T)}| K_{i+1}=k]} \\
&=\frac{\Eb[X_{i+1} \lv_{i\leq N(T)}| K_{i+1}=k]}{\Eb[ \lv_{i\leq N(T)}]}\label{2eqn:hatX}  ,
\end{align}
where the last equality follows from the fact that the two events \tcb{$i\leq N(T)$} and $K_{i+1}=k$ are independent \tcb{of} each other under any online policy in $\Pi_5$.

Taking expectation on both sides of $(\ref{2eqn:hatX})$ with respect to $k$, we have
\begin{align}\label{2eqn:hatX_1order}
\Eb_k[\hat{X}_{i+1}(k)] \cdot \Eb[ \lv_{i\leq N(T)}] = \Eb[X_{i+1} \lv_{i\leq N(T)}]   .
\end{align}
Meanwhile, we note that
\begin{align}
&\left(\hat{X}_{i+1}(k) {\Eb[ \lv_{i\leq N(T)}]}\right)^2
=\left(\Eb\left[X_{i+1}  \lv_{i\leq N(T)}\middle| K_{i+1}=k\right]\right)^2\\
& \leq \Eb\left[X^2_{i+1} \lv_{i\leq N(T)}\middle| K_{i+1}=k\right]\Eb[ \lv_{i\leq N(T)}| K_{i+1}=k]  \label{2eqn:cauchy-1} \\
&=  \Eb\left[X^2_{i+1} \lv_{i\leq N(T)}\middle| K_{i+1}=k\right]\Eb[ \lv_{i\leq N(T)}]\label{2eqn:cauchy-2}  ,
\end{align}
where (\ref{2eqn:cauchy-1}) follows from the Cauchy-Schwartz inequality.
Dividing both sides of (\ref{2eqn:cauchy-2}) by $\Eb[ \lv_{i\leq N(T)}]$ and taking expectation with respect to $k$, we have
\begin{align}
\Eb_k\left[\hat{X}^2_{i+1}(k)\right] \Eb[ \lv_{i\leq N(T)}] & \leq \Eb\left[X^2_{i+1} \lv_{i\leq N(T)}\right]\label{2eqn:hatX_2order}    .
\end{align}

Next, based on Lemma~\ref{2lemma:extra-2}, we have
\begin{align}
&\lim_{T\rightarrow \infty} \frac{\Eb\left[R(T) \right] }{T}=\lim_{T\rightarrow \infty} \frac{\Eb\left[\sum_{i=1}^{N(T)+1} X_i^2 \right] }{2T}
\\
&= \lim_{T\rightarrow \infty} \frac{\sum_{i=0}^{\infty} \Eb\left[ X_{i+1}^2\lv_{i\leq N(T)}\right]}{2T}\\
&\geq \lim_{T\rightarrow \infty}\frac{  \Eb\left[\sum_{i=0}^{\infty} X_{i+1}^2 \lv_{i\leq N(T)}\right]}{2 \Eb[\sum_{i=0}^{\infty} X_{i+1} \lv_{i\leq N(T)}]} \\
&\geq \lim_{T\rightarrow \infty} \frac{\sum_{i=0}^{\infty}\Eb_k [\hat{X}^2_{i+1}(k)] \cdot \Eb[ \lv_{i\leq N(T)}] }{2 \sum_{i=0}^{\infty} \Eb_k[\hat{X}_{i+1}(k)] \cdot \Eb[ \lv_{i\leq N(T)}]}\label{2eqn:hatX2}   ,
\end{align}
where in (\ref{2eqn:hatX2}) the first inequality follows from the fact that $T\leq \sum_{i=0}^{\infty} X_{i+1} \lv_{i\leq N(T)}$ for every sample path, and the second inequality follows from (\ref{2eqn:hatX_1order}) and (\ref{2eqn:hatX_2order}).

Define 
\begin{align}\label{2eqn:rho}
\rho_{i+1}&:= \frac{\Eb_k[\hat{X}_{i+1}(k)] \cdot \Eb[ \lv_{i\leq N(T)}]}{ \sum_{i=0}^{\infty} \Eb_k[\hat{X}_{i+1}(k)] \cdot \Eb[ \lv_{i\leq N(T)}]}.
\end{align}
We note that $\{\rho_{i+1}\}_{i=0}^\infty$ is a valid distribution. Therefore, based on Cauchy-Schwartz inequality, we have
\begin{align} \label{2eqn:cauchy}
&\left(\sum_{i=0}^{\infty}\frac{\hat{X}^2_{i+1}(k)}{\Eb[\hat{X}_{i+1}(k)]} \rho_{i+1}\right) \left(\sum_{i=0}^{\infty}\Eb[\hat{X}_{i+1}(k)] \rho_{i+1}\right)  \nonumber\\
&\geq \left(\sum_{i=0}^{\infty} \hat{X}_{i+1}(k) \rho_{i+1}\right)^2 := \left(\bar{X}(k)\right)^2 ,
\end{align}
where $\bar{X}(k):=\sum_{i=0}^{\infty}\hat{X}_{i+1}(k) \rho_{i+1}$. This is equivalent to
\begin{align}
\sum_{i=0}^{\infty}\frac{\hat{X}_{i+1}^2(k)}{\Eb[\hat{X}_{i+1}(k)]}\rho_{i+1}
\geq
\frac{\left(\bar{X}(k)\right)^2}{\sum_{i=0}^{\infty}\Eb[\hat{X}_{i+1}(k)] \rho_{i+1}}
=
\frac{\left(\bar{X}(k)\right)^2}{\Eb[\bar{X}(k)]}. \label{2eqn:add-2}
\end{align}
We note that (\ref{2eqn:hatX2}) equals $\lim_{T\rightarrow \infty} \sum_{i=0}^{\infty}\frac{\Eb_k [\hat{X}^2_{i+1}(k)]  }{2  \Eb_k[\hat{X}_{i+1}(k)] } \rho_{i+1} $,
which is lower bounded by $\lim_{T\rightarrow \infty} \frac{\Eb[\bar{X}^2(k)]}{2\Eb[\bar{X}(k)]}$ according to (\ref{2eqn:add-2}).

Before we proceed to define the renewal policy, we will first show that $\bar{X}(k+1)-\bar{X}(k)\geq 0$ for $k=1,2,\ldots$. Consider the $(i+1)$st inter-update delay $X_{i+1}$ where $i\leq N(T)$. 
Group all sample paths that share the same history up to the $k$th attempt together. Depending on whether the $k$th attempt is successful, we can further divide them into two subgroups. Then, all those who fail at the $k$th attempt will experience longer inter-update delay than those who succeed at the $k$th attempt. Since each attempt is successful with probability $p$ independently, after taking expectation over all such sample paths we must have
\begin{align}
\hat{X}_{i+1}(k+1)\geq\hat{X}_{i+1}(k) .  \label{2eqn:add-1}
\end{align}
From (\ref{2eqn:add-1}) and the definition of $\bar{X}(k)$, we then have $\bar{X}(k+1)-\bar{X}(k)\geq 0$ for $k=1,2,\ldots$.

Then, we define the a renewal policy as follows: Starting at $t=0$, the sensor will first update at time $\bar{X}(1)$ and observe the feedback. If the update is successful, the sensor will wait for $\bar{X}(0)$ and update again; Otherwise, it will update again after waiting for $\bar{X}(2)-\bar{X}(1)$. The process continues after waiting for $\bar{X}(k+1)-\bar{X}(k)$, where $k-1$ is the number of {\it failed} updated since the last successful update.

Define
$q_i :=\frac{\Eb[\lv_{i\leq N(T)}]}{\Eb[N(T)+1]}   .$
We note that $\sum_{i=0}^{\infty}q_i=1$, thus $\{q_i\}_{i=0}^{\infty}$ is a valid distribution. 

Based on the definitions of $\bar{X}(k)$ and $\rho_{i+1}$, we have
\begin{align}
&\Eb_k[\bar{X}(k)]\nonumber \\
&=\Eb_k\left[ \sum_{i=0}^{\infty}\hat{X}_{i+1}(k) \frac{\Eb_k[\hat{X}_{i+1}(k)] \cdot \Eb[ \lv_{i\leq N(T)}]}{ \sum_{i=0}^{\infty} \Eb_k[\hat{X}_{i+1}(k)] \cdot \Eb[ \lv_{i\leq N(T)}]} \right]   \\
&=\frac{\sum_{i=0}^{\infty}\Eb_k^2 [\hat{X}_{i+1}(k)] \cdot \Eb[\lv_{i\leq N(T)}]}{\sum_{i=0}^{\infty}\Eb_k [\hat{X}_{i+1}(k)] \cdot \Eb[\lv_{i\leq N(T)}]} \\
&=\frac{\sum_{i=0}^{\infty}q_i\Eb^2[\hat{X}_{i+1}] 
}{\sum_{i=0}^{\infty}q_i\Eb[\hat{X}_{i+1}]   }  \geq \frac{\left(\sum_{i=0}^{\infty}q_i\Eb[\hat{X}_{i+1}]\right)^2}{\sum_{i=0}^{\infty}q_i\Eb[\hat{X}_{i+1}]}
\label{2eqn:renew-31}  \\
&=\sum_{i=0}^{\infty}q_i\Eb[\hat{X}_{i+1}] =\sum_{i=0}^{\infty} \Eb[\hat{X}_{i+1}] \frac{\Eb[\lv_{i\leq N(T)}]}{\Eb[N(T)+1]}    \\
&=\frac{\sum_{i=0}^{\infty} \Eb[X_{i+1} \lv_{i\leq N(T)}]}{\Eb[N(T)+1]}  \geq \frac{T}{\Eb[N(T)+1]}  \label{2eqn:renew-4}   
\end{align}
where (\ref{2eqn:renew-31}) follows from Jensen's inequality and (\ref{2eqn:renew-4}) follows from (\ref{2eqn:hatX_1order}).

Let $\bar{N}(T)$ denote the number of completed renewal intervals under policy $\{\bar{X}(k)\}$ by time $T$.  Then, according to the elementary renewal theorem~\cite{ross:1996}, 
\begin{align*}
\lim_{T\rightarrow \infty} \frac{\Eb[\bar{N}(T)]}{T} = \lim_{T\rightarrow\infty} \frac{1}{\Eb_k[\bar{X}(k)]}\leq  \lim_{T\rightarrow\infty}\frac{\Eb[N(T)+1]}{T} \leq p.
\end{align*}
Therefore, for any $\pi\in\Pi_5$, we can always construct a renewal policy that is also in $\Pi_5$, and achieves a shorter long-term average AoI.

\bibliographystyle{IEEEtran}
\bibliography{AgeInfo,ener_harv}

\end{document}